\documentclass[11pt]{article}
\usepackage{amsmath,amsthm,amsfonts,amssymb}
\usepackage{fullpage}
\usepackage{hyperref}
\usepackage{url}
\usepackage{amscd,amssymb,amsmath,amsrefs,amsthm,latexsym,enumerate}
\usepackage{tikz}
\usetikzlibrary{quantikz}
\usepackage{tikz-network}
\usepackage{array}
\newcommand{\nums}{\mathbb{N}}
\newcommand{\ints}{\mathbb{Z}}
\newcommand{\oddints}{\ints_{\textit{odd}}}
\newcommand{\rats}{\mathbb{Q}}
\newcommand{\oddrats}{{\rats_{\textit{odd}}}}
\newcommand{\reals}{\mathbb{R}}

\newcommand{\lcm}{\textup{lcm}}
\newcommand{\floor}[1]{{\lfloor{#1}\rfloor}}

\newcommand{\ordth}[1]{{{#1}^\textit{th}}}

\newcommand{\Mod}{\textup{Mod}}
\newcommand{\CNOT}{\textup{C-NOT}}
\newcommand{\map}[3]{{{#1}:{#2}\rightarrow{#3}}}
\newcommand{\scong}[1]{\mathrel{\approx_{#1}}}
\newcommand{\myctrl}[1]{{\textup{C-}{#1}}}

\newcommand{\eref}[1]{Eq.~(\ref{#1})}
\newcommand{\erefs}[1]{Eqs.~(\ref{#1})}
\theoremstyle{plain}
\newtheorem{Thm}{Theorem}[section]
\newtheorem{Cor}[Thm]{Corollary}
\newtheorem{Lem}[Thm]{Lemma}
\newtheorem{Prp}[Thm]{Proposition}
\newtheorem{Fact}[Thm]{Fact}
\newtheorem{Obs}[Thm]{Observation}
\theoremstyle{definition}
\newtheorem{Def}[Thm]{Definition}
\newtheorem{remark}[Thm]{Remark}

\title{Implementing the quantum fanout operation with simple pairwise interactions}

\author{Stephen Fenner\thanks{Computer Science and Engineering Department, Columbia, SC 29208.  \url{fenner@cse.sc.edu}, \url{rwosti@email.sc.edu}} \\ University of South Carolina
\and
Rabins Wosti$\!\,^*$ \\ University of South Carolina}

\begin{document}
\maketitle

\begin{abstract}
It has been shown that, for even $n$, evolving $n$ qubits according to a Hamiltonian that is the sum of pairwise interactions between the particles, can be used to exactly implement an $(n+1)$-qubit fanout gate using a particular constant-depth circuit~[\href{https://arXiv.org/abs/quant-ph/0309163}{arXiv:quant-ph/0309163}]. However, the coupling coefficients in the Hamiltonian considered in that paper are assumed to be all equal. In this paper, we generalize these results and show that for all $n$, including odd $n$, one can exactly implement an $(n+1)$-qubit parity gate and hence, equivalently in constant depth an $(n+1)$-qubit fanout gate, using a similar Hamiltonian but with unequal couplings, and we give an exact characterization of which couplings are adequate to implement fanout via the same circuit.

We also investigate pairwise couplings that satisfy an inverse square law, giving necessary and sufficient criteria for implementing fanout given spatial arrangements of identical qubits in two and three dimensions subject to this law.  We use our criteria to give planar arrangements of four qubits that (together with a target qubit) are adequate to implement $5$-qubit fanout.
\end{abstract}


\section{Introduction}

\subsection{Previous work}

In the study of classical Boolean circuit complexity, the fanout operation---where a Boolean value on a single wire is copied into any number of wires---is taken for granted as cost-free.  The picture is very different, however, with quantum circuits with unitary gates, where the number of wires is fixed throughout the circuit.  There, fanout gates are known to be very powerful primitives for making shallow quantum circuits~\cites{GHMP:QACC,HS:fanout,Moore:fanout,Spalek:fanout}.  It has been shown that in the quantum realm, fanout, parity (see below), and $\Mod_q$ gates (for any $q \ge 2$) are all equivalent up to constant depth and polynomial size \cites{GHMP:QACC,Moore:fanout}.  That is, each gate above can be simulated exactly by a constant-depth, polynomial-size quantum circuit using any of the other gates above, together with standard one- and two-qubit gates (e.g., $\CNOT$, $H$, and $T$).  This is not true in the classical case, where, for example, parity cannot be computed by constant-depth Boolean circuits with fanout and unbounded AND-, OR-, and NOT-gates~\cites{Ajtai:AC0,FSS:AC0,Hastad:AC0}.  Furthermore, using fanout gates, in constant depth and polynomial size one can approximate sorting, arithmetical operations, phase estimation, and the quantum Fourier transform~\cites{HS:fanout,Spalek:fanout}.  Fanout gates can also exactly implement $n$-qubit threshold gates, unbounded AND-gates (generalized Toffoli gates), and OR-gates in constant depth~\cites{TT:constant-depth-collapse}.  Since long quantum computations may be difficult to maintain due to decoherence, shallow quantum circuits may prove much more realistic, at least in the short term, and finding ways to implement fanout would then lend enormous power to these circuits.

On the negative side, fanout gates so far appear hard to implement by traditional quantum circuits.  There is mounting theoretical evidence that fanout gates cannot be simulated in small (sublogarithmic\footnote{Fanout on $n$ qubits can be implemented by a $O(\log n)$-depth circuit with $O(n)$ many $\CNOT$ gates.}) depth and small width, even if unbounded AND-gates
are allowed~\cites{FFGHZ:fanout,Rosenthal:parity}.

Therefore, rather than trying to implement fanout with a traditional small-depth quantum circuit, an alternate approach would be to evolve an $n$-qubit system according to one or more (hopefully implementable) Hamiltonians, along with a minimal number of traditional quantum gates.  It was shown in~\cites{Fenner:fanout,FZ:heisenberg} that simple Hamiltonians using spin-exchange (Heisenberg) interactions
do exactly this.  Those papers presented a simple quantum circuit for computing $n$-bit parity (equivalent to fanout) that included two invocations of the Hamiltonian along with a constant number of one- and two-qubit Clifford gates.

More recently, Guo et al.~\cite{GuoEtAl:fanout} presented a method for implementing fanout on a mesh of qubits.  Their approach involves a series of modulated long-range Hamiltonians applied to the qubits obeying inverse power laws.

\subsection{The current work}

This paper revisits the spin-exchange Hamiltonians considered in~\cites{Fenner:fanout,FZ:heisenberg}.  A major weakness of that work is that it assumes all the pairwise couplings between the spins to be equal.  This is physically unrealistic since we expect couplings between spins that are spatially far apart to be weaker than those between spins in close proximity.

In this paper, we show that $n$-qubit fanout can still be implemented by the exact same circuit $C_n$ given in~\cite{Fenner:fanout}, even with a wide variety of unequal pairwise couplings.  We also give an exact characterization of which couplings are allowed so that $C_n$ implements fanout.

Formally, the $n$-qubit fanout gate is the $(n+1)$-qubit unitary operator $F_n$ is defined such that $F_n\ket{x_1,\ldots, x_n,c} = \ket{x_1\oplus c,\ldots,x_n\oplus c,c}$ for all $x_1,\ldots,x_n,c\in\{0,1\}$.  The $n$-qubit parity gate is the $(n+1)$-qubit unitary operator $P_n$ such that
$P_n\ket{x_1,\ldots, x_n,t} = \ket{x_1,\ldots,x_n,t\oplus x_1\oplus\cdots\oplus x_n}$ for all $x_1,\ldots,x_n,t\in\{0,1\}$.  It was shown in \cite{Moore:fanout} that $F_n = H^{\otimes(n+1)}P_n H^{\otimes(n+1)}$, where $H$ is the $1$-qubit Hadamard gate.  Thus $F_n$ and $P_n$ are equivalent in constant depth, and any circuit implementing $P_n$ can be converted to one implementing $F_n$ by conjugating with a bank of Hadamard gates.

The circuit $C_n$ implements $P_n$ and is shown in Figure~\ref{fig:parity-circuit}.
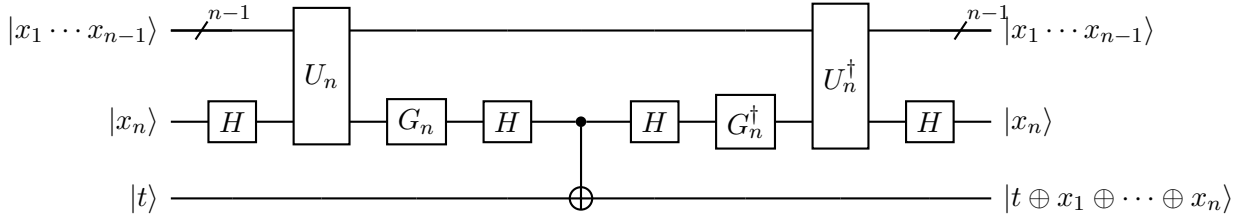
\begin{figure}
\begin{quantikz}
\lstick{\ket{x_1\cdots x_{n-1}}} & \qw\qwbundle{n-1} & \gate[wires=2]{U_n} &\qw &\qw &\qw &\qw &\qw & \gate[wires=2]{U_n^\dagger} &\qw & \rstick{\ket{x_1\cdots x_{n-1}}}\qwbundle{n-1}\qw & \\
\lstick{\ket{x_n}} & \gate{H} & & \gate{G_n} & \gate{H} & \ctrl{1} & \gate{H} & \gate{G_n^\dagger} & & \gate{H} & \rstick{\ket{x_n}}\qw \\
\lstick{\ket{t}} &\qw&\qw&\qw&\qw& \targ{} &\qw&\qw&\qw&\qw& \rstick{\ket{t\oplus x_1\oplus \cdots \oplus x_n}}\qw
\end{quantikz}
\caption{The circuit $C_n$ implements the parity gate $P_n$.  It uses the unitary operator $U_n$ and its adjoint once each.  Here, $G_n = S^{1-n}$ is either $S$ (the phase gate), $I$, $S^\dagger$, or $Z$ (the Pauli $z$-gate) if $n$ is congruent to $0$, $1$, $2$, or $3$, respectively.  Since $U_n$ is swap-invariant, the single-qubit gates can be moved to any of the first $n$ qubits, together with the control of the $\CNOT$ gate.}\label{fig:parity-circuit}
\end{figure}
Here, the $1$-qubit Clifford gate $G_n$ is either $S$, $I$, $S^\dagger$, or $Z$, depending on $n\bmod 4$, where $I$ is the identity, $S$ satisfies $S\ket{b} = i^b\ket{b}$ for $b\in\{0,1\}$, and $Z$ is the Pauli $z$-gate.  The unitary operator $U_n$ is defined as follows: for all $x = x_1\cdots x_n\in\{0,1\}^n$, letting $w = x_1+\cdots+x_n$,
\begin{equation}\label{eqn:U-n}
     U_n\ket{x} = i^{w(n-w)}\ket{x}\;.
\end{equation}
It was shown in~\cite{Fenner:fanout} that $U_n$ is the result of running a particular Hamiltonian $H_n$, defined below, for a certain amount of time on the first $n$ qubits.  It also was shown that $C_n = P_n$ for even $n$, and a similar calculation shows the same is true for odd $n$.  For the full result and its proof, see Appendix~\ref{sec:parity}.

We consider Hamiltonians of the form $H_n = \sum_{1\le i<j\le n} J_{i,j} Z_iZ_j$, where $Z_i$ and $Z_j$ are Pauli $Z$-gates acting on the $\ordth{i}$ and $\ordth{j}$ qubits, respectively, and the $J_{i,j}$ are real coupling constants (in units of energy).  $H_n$ is a simplified version of the spin-exchange interaction where only the $z$-components of the spins are coupled.  It bears some resemblance to a quantum version of the Ising model, as described in~\cite{Wikipedia:Transverse-field_Ising_model}, but with no transverse field and allowing long-range as well as nearest-neighbor couplings.  In~\cite{Fenner:fanout} it was shown that $U_n = e^{-iH_nt}$ for a certain time $t$, provided all the coupling constants $J_{i,j}$ are equal.

In this paper, we exactly characterize when $H_n$ can be run to implement $U_n$ by proving the following result in Section~\ref{sec:main}:

\begin{Thm}\label{thm:main}
$U_n \propto e^{-iH_nt}$ for some $t>0$ if and only if there exists a constant $J>0$ such that (1.) all $J_{i,j}$ are odd integer multiples of $J$, and (2.) the graph $G$ on vertices $1,\ldots,n$ with edge set $\{\{i,j\} : \mbox{$i<j$ and $J_{i,j}/J \equiv 3 \pmod{4}$}\}$ is Eulerian\footnote{We use this term in the looser sense that the graph need not be connected.}, that is, all its vertices have even degree.

Furthermore, if $t$ exists, we can set $t := \pi\hbar/4J$.
\end{Thm}

Our result gives more flexibility in the coupling constants, allowing stronger and weaker couplings for spins placed nearer and farther apart, respectively.  For example, suppose we have four identical spins arranged in the corners of a square.  The spins diagonally opposite each other may have coupling constant $J$ whereas neighboring spins can have coupling constant $3J$.  The corresponding $g_{i,j} = 1$ is thus odd for neighboring spins, but this arrangement can be used to implement $U_4$, because the edges connecting neighboring spins form a square, which is Eulerian. For the spins arranged in the corners of a regular cube, neighboring spins may have coupling constant $7J$, spins on the diagonal ends of each face may have coupling constant $3J$ and the antipodal spins may have coupling constant $J$. Thus, the corresponding $g_{i,j}$ for antipodal spins is even while it is odd for the neighboring and diagonal spins and therefore, this arrangement can be used to implement $U_8$ because the edges connecting the neighboring spins and the spins on the diagonal ends of each face of a regular cube form an Eulerian graph. Similarly, for spins arranged on the corners of a regular octahedron, the graph of neighboring spins is Eulerian, so neighboring spins can have coupling $3J$ and antipodal spins $J$.

Starting in Section~\ref{sec:inverse-square-law}, we also investigate how spatial arrangements of qubits whose couplings obey an inverse square law can be used to implement $U_n$ exactly.  We find severe limitations on such arrangements.  In particular, we show that no three qubits can lie on the same line, and this generally rules out any kind of mesh arrangement.  Such arrangements therefore cannot implement $U_n$, assuming an inverse square law, unless extra physical barriers are used to moderate the couplings between certain pairs of qubits.  We give complete characterizations of which spatial arrangements in $\reals^2$ and $\reals^3$ can implement $U_n$ exactly.

Our work differs from the recent work of Guo et al.~\cite{GuoEtAl:fanout} in a number of respects.  They adapt a state transfer protocol of Eldredge et al.~\cite{PhysRevLett.119.170503} that, given an arbitrary $1$-qubit state $\alpha\ket{0}+\beta\ket{1}$, produces the GHZ-like state $\alpha\ket{0\cdots 0}+\beta\ket{1\cdots 1}$ on $n$ qubits.  Their protocol uses long-range interactions on a mesh of qubits by sequentially turning on and off various Hamiltonians to implement a cascade of C-NOT gates, where different Hamiltonians must be applied at different times.  Our scheme runs a simple, swap-invariant Hamiltonian twice, together with a constant number of $1$-qubit gates and a C-NOT gate connecting to the target.  Unlike in~\cite{GuoEtAl:fanout}, our scheme needs no ancilla qubits.  If the pairwise couplings must satisfy an inverse-square law, however, then our scheme has the disadvantages described above.





\section{Preliminaries}\label{sec:prelims}

We let $\ints$ denote the set of integers and $\nums$ the set of nonnegative integers.  We choose physical units so that $\hbar = 1$. For $n\in\nums$ and bit vector $x\in\{0,1\}^n$, we let $w(x)$ denote the Hamming weight of $x$, and we let $x_i$ denote the $\ordth{i}$ bit of $x$, for $1\le i\le n$.  We use $[n]$ to denote the set $\{1, \ldots, n\}$.  For $x,y,\alpha\in\reals$ with $\alpha > 0$, we write $x\equiv_\alpha y$ to mean that $(x-y)/\alpha$ is an integer, and we let $x\bmod \alpha$ denote the unique $y\in [0,\alpha)$ such that $x\equiv_\alpha y$.  For bits $a,b\in\{0,1\}$ we write $a\oplus b$ to mean $(a+b)\bmod 2$.  For vectors or operators $U$ and $V$ of the same type, we write $U\propto V$ to mean there exists $\theta\in\reals$ such that $U = e^{i\theta}V$, i.e., $U$ and $V$ differ by a global phase factor.

We use the symbol `$:=$' to mean ``equals by definition.''

\section{Main Results}\label{sec:main}




We consider a particular type of Hamiltonian $H_n$, acting on a system of $n \in \nums$ qubits, as the weighted sum of pairwise $Z$-interactions among the qubits in analogy to spin-exchange (Heisenberg) interactions:
\begin{equation}\label{eqn:Hamiltonian}
H_n := \sum_{1\le i<j\le n} J_{i,j} Z_i Z_j\;,
\end{equation}
where $Z_k$ is the Pauli $Z$-gate acting on the $\ordth k$ qubit for $k \in [n]$, and for $1 \leq i<j \leq n$, \ $J_{i,j}\in \reals$ is the coupling coefficient between the $\ordth i$ and $\ordth j$ qubits.  For convenience, we define $J_{j,i} := J_{i,j}$ for all $1\le i<j\le n$.
  $H_n$ differs from the usual (isotropic) Heisenberg interactions in that only the $z$-components of the spins are coupled.

Let $x=x_1 \cdots x_n\in\{0,1\}^n$ be a vector of $n$ bits, where each $x_i$ denotes the $\ordth{i}$ bit of $x$. Notice that $Z_iZ_j\ket{x}=(-1)^{x_i+x_j}\ket{x}$ for $1\leq i<j \le n$, that is, $Z_iZ_j$ flips the sign of $\ket{x}$ iff $x_i \ne x_j$. Further, for $t,\theta\in\reals$, let
\begin{equation}\label{eqn:V-n}
V_n := V_n(t,\theta) :=  e^{-i\theta}e^{-iH_nt}
\end{equation}
be the unitary operator realized by evolving the Hamiltonian $H_n$ of \eref{eqn:Hamiltonian} for time period $t$, where $\theta$ represents a global phase factor that may be introduced into the system. It has been explicitly shown in~\cite{Fenner:fanout} that for $n \equiv_4 2$, if $V_n\propto U_n$ (see \eref{eqn:U-n}),
one can realize the parity gate $P_n$ (and thus the fanout gate $F_n$) in constant additional depth for $n$ qubits via the quantum circuit $C_n$ shown in Figure~\ref{fig:parity-circuit}.  This fact indeed holds for all $n$, and we give a unified proof in Appendix~\ref{sec:parity} that the circuit of Figure~\ref{fig:parity-circuit} works for all $n$.  Further, it was shown in the same paper that $V_n\propto U_n$ if all the $J_{i,j}$ are equal, and we give an updated proof of this in Appendix~\ref{sec:U-n}, where we prove the following:

\begin{Lem}\label{lem:U-n}
For $n\ge 1$, let $H_n := J\sum_{1\le i<j\le n} Z_i Z_j$ for some $J>0$.  Then $U_n = V_n(t,\theta)$ for some $\theta\in\reals$, where $t := \pi/(4J)$ and $V_n(t,\theta)$ is as in \eref{eqn:V-n}.\footnote{$J$ is in units of energy and $t$ is in units of time, but this fact is irrelevant to our results; one can assume that $J$ and $t$ are unitless quantities.  In any case, $Jt$ is unitless, as we are taking $\hbar := 1$.}
\end{Lem}

\begin{proof}
See Appendix~\ref{sec:U-n}.
\end{proof}

The two main goals of the current work are (1) to show that equality of the $J_{i,j}$ is not necessary and (2) to determine exactly for which values of $J_{i,j}$ this is possible.  We will use Lemma~\ref{lem:U-n} to establish Theorem~\ref{thm:main}, the proof of which is the goal of this section.

Let $H_n$ be as in \eref{eqn:Hamiltonian} for arbitrary $J_{i,j}$.  For $x\in\{0,1\}^n$ and $t,\theta_1\in\reals$, setting $k_{i,j} := J_{i,j}t$ for convenience, we have
\begin{align}\label{eqn:V-n-of-x}
V_n(t,\theta_1)\ket{x}
& = \exp\left(-i\theta_1 - i\sum_{1\leq i<j \leq n} k_{i,j} (-1)^{x_i + x_j}\right)\ket{x}\;.
\end{align}
Using the fact that $U_n\ket{x} = \exp\left(i(\pi/2)w(x)(n-w(x))\right)\ket{x}$ and equating exponents, the condition that $V_n(t,\theta_1) = U_n$ is seen to be equivalent to
\begin{equation}
\theta_1 + \sum_{1 \leq i<j \leq n} k_{i,j} (-1)^{x_i + x_j} \equiv_{2\pi} -\left(\frac{\pi}{2}\right)w(x)(n-w(x)) \label{PhaseEqn}
\end{equation}
holding for all $x = x_1\cdots x_n\in\{0,1\}^n$.
Lemma~\ref{lem:U-n} yields a similar phase congruence in the case where $k_{i,j} = Jt = \pi/4$ for all $i<j$: there exists $\theta_2\in\reals$ such that for all $x\in\{0,1\}^n$,
\begin{equation}\label{eqn:U-n-phase}
\theta_2 + \frac{\pi}{4}\sum_{1\le i<j\le n} (-1)^{x_i+x_j} \equiv_{2\pi} -\left(\frac{\pi}{2}\right)w(x)(n-w(x))\;.
\end{equation}
Subtracting \eref{eqn:U-n-phase} from \eref{PhaseEqn} and rearranging, we get that $V_n(t,\theta) = U_n$ is equivalent to
\begin{align*}
\sum_{1 \leq i<j \leq n} \left(k_{i,j}-\frac{\pi}{4}\right) (-1)^{x_i + x_j} &\equiv_{2\pi} \theta_2 - \theta_1 & \forall x \in \{0,1\}^n\;,
\end{align*}
or equivalently, setting $f_{i,j} := k_{i,j} - \pi/4$ for all $1\le i<j\le n$,
\begin{align}
\sum_{1 \leq i<j \leq n} f_{i,j} (-1)^{x_i + x_j} &\equiv_{2\pi} \theta_2 - \theta_1 & \forall x \in \{0,1\}^n\;. \label{fijEqn}
\end{align}
Substituting the zero vector for $x$ in \eref{fijEqn} implies $\theta_2-\theta_1 \equiv_{2\pi} \sum_{i<j} f_{i,j}$, so \eref{fijEqn} can be rewritten as 
\begin{align} \nonumber
\sum_{i<j} f_{i,j} (-1)^{x_i + x_j} &\equiv_{2\pi} \sum_{i<j} f_{i,j} \\ \nonumber
\sum_{i<j} f_{i,j} \left((-1)^{x_i + x_j} - 1 \right) &\equiv_{2\pi} 0 \\
\sum_{i<j\;:\;x_i \neq x_j} f_{i,j} &\equiv_{\pi} 0 & \forall x\in\{0,1\}^n\;. \label{withoutsymm}
\end{align}
We have thus established the following lemma:

\begin{Lem} \label{lem:fij-meaning}
Let $H_n$ be as in (\ref{eqn:Hamiltonian}) and let $t\in\reals$ be arbitrary.  There exists $\theta\in\reals$ such that $V_n(t,\theta) = U_n$, if and only if \eref{withoutsymm} holds, where $f_{i,j} := J_{i,j}t - \pi/4$ for all $1\le i<j\le n$.
\end{Lem}


\begin{Lem} \label{lem:discrete}
Let $\{f_{i,j}\}_{1\le i<j\le n}$ be real numbers such that \eref{withoutsymm} holds.  Then $f_{i,j} \equiv_{\pi/2} 0$ for all $1 \leq i < j \leq n$.
\end{Lem}

\begin{proof}
For convenience, define $f_{j,i} := f_{i,j}$ for all $i<j$.  For $a\in [n]$, let $x^{(a)}\in\{0,1\}^n$ be the $n$-bit vector whose $\ordth{a}$ bit is $1$ and whose other bits are all $0$.  Consider two different bit vectors $x^{(a)}$ and $x^{(b)} \in \{0,1\}^n$ for $a<b$.  Also, consider a third bit vector $y \in \{0,1\}^n$ with $w(y)=2$ such that its bits are set to 1 in exactly the $a$ and $b$ positions, i.e., $y = x^{(a)}\oplus x^{(b)}$.  Plugging in $x^{(a)}$, $x^{(b)}$, and $y$, respectively into \eref{withoutsymm}, we have
\begin{align}
\sum_{j \in [n]\;:\; j \neq a} f_{a,j} &\equiv_\pi 0 \label{aEqn} \\
\sum_{i \in [n]\;:\; i \neq b} f_{i,b} &\equiv_\pi 0 \label{bEqn} \\
\sum_{k \in [n]\;:\; k \notin \{a,b\}} (f_{a,k} + f_{k,b}) &\equiv_\pi 0 \label{kEqn}
\end{align}
\eref{aEqn}$+$(\ref{bEqn})$-$(\ref{kEqn}) gives
\begin{equation}
\left(\sum_{j \in [n]\;:\; j \neq a} f_{a,j} - \sum_{k \in [n]\;:\; k \notin \{a,b\}} f_{a,k} \right) + \left(\sum_{i \in [n]\;:\; i \neq b} f_{i,b} - \sum_{k \in [n]: k \notin \{a,b\}} f_{k,b} \right) = 2f_{a,b} \equiv_\pi 0\;.
\end{equation}
Therefore, $f_{a,b} \equiv_{\pi/2} 0$.
Since, $a$ and $b$ are chosen arbitrarily, the conclusion follows.
\end{proof}

\begin{Def} \label{DefM_n}
For $n \geq 2$, let $M_n$ be the $2^n\times \binom{n}{2}$ matrix over the $2$-element field $\mathbb{F}_2$ with rows $m_x$ indexed by bit vectors $x$ of length $n$ and columns indexed by pairs $\{i,j\}$ for $1\le i < j \le n$, whose $\ordth{(x,\{i,j\})}$ entry is $m_{x,\{i,j\}} = x_i \oplus x_j$.
\end{Def}

\begin{Lem} \label{LemM_n}
Every matrix $M_n$ defined by Definition~\ref{DefM_n} has rank $n-1$, and its rows are spanned by any set of $n-1$ rows $m_x$ for $x$ with Hamming weight~1.
\end{Lem}
\begin{proof}
All scalar and vector addition below is over $\mathbb{F}_2$.
Let $S:=\{x \in \{0,1\}^n: w(x)=1\}$ be the set of $n$-bit vectors of Hamming weight~1.  For $n$-bit vectors $r$ and $s$, we can write the $\ordth{\{i,j\}}$ component of the sum $m_r+m_s$ as
\begin{equation*}
(m_r + m_s)_{\{i,j\}} = m_{r,\{i,j\}} + m_{s,\{i,j\}} = (r_i + r_j) + (s_i + s_j) = (r_i + s_i) + (r_j + s_j) \\
= m_{r + s,\{i,j\}}\;,
\end{equation*}
and thus $m_r + m_s = m_{r+s}$.
With this observation, we can infer that every row in the matrix $M_n$
can be expressed as the sum of the rows indexed by $n$-bit vectors in $S$. In particular, we have
$$
\sum_{x \in S} m_x = m_{11\cdots 1} = \vec{0}\;.
$$
This causes a linear dependence among the $n$ vectors in the set $S$.  The sum of rows indexed by any nonempty proper subset of $S$, however, results in a row indexed by an $n$-bit vector containing at least one $0$ and one $1$, and thus cannot be all zeros, which means there is no linear dependence corresponding to any proper subset of $S$.  It follows that every matrix $M_n$ of the above form has rank $n-1$, and any set of $n-1$ rows with indices in $S$ spans all the rows of $M_n$.
\end{proof}

Notice that Lemma~\ref{lem:discrete} results in the following corollary as an immediate consequence.
\begin{Cor} \label{1Corollary}
Let $\{f_{i,j}\}_{1\le i<j\le n}$ be as in Lemma~\ref{lem:discrete}, and define $g_{i,j} := 2f_{i,j}/\pi$ for all $1\le i<j\le n$.  Then $g_{i,j}\in\ints$ for all $i<j$, and \eref{withoutsymm} is equivalent to $M_ng \equiv_2 \vec 0$, where $g$ is the column vector with entries $g_{i,j}$.
\end{Cor}

\begin{proof}[Proof of Theorem~\ref{thm:main}]
Let $H_n$ be as in \eref{eqn:Hamiltonian}.  For $t>0$, the statement that $U_n \propto e^{-iH_nt}$ is equivalent to the existence of some $\theta\in\reals$ such that $V_n(t,\theta) = U_n$, where $V_n(t,\theta)$ is defined by \eref{eqn:V-n}.  By Lemma~\ref{lem:fij-meaning}, this in turn is equivalent to \eref{withoutsymm}, i.e., $\sum_{i < j\;:\;x_i \neq x_j} f_{i,j} \equiv_\pi 0$ for all $n$-bit vectors $x$, where $f_{i,j} := J_{i,j}t-\pi/4$ for all $1\le i < j \le n$.
From Lemma~\ref{lem:discrete} and Corollary~\ref{1Corollary},
\eref{withoutsymm} holds if and only if
\begin{enumerate}[(i)]
\item
$f_{i,j} \equiv_{\pi/2} 0$, and therefore, letting $g_{i,j}:=2f_{i,j}/\pi$, we have $g_{i,j}\in\ints$ for all $1 \leq i < j \leq n$, and
\item
$M_ng \equiv_2 \vec{0}$, where $g$ is the $
\binom {n}{2}$-dimensional column vector of $g_{i,j}$'s and $M_n$ is as in Definition~\ref{DefM_n}.
\end{enumerate}

Solving for $J_{i,j}$ in terms of $f_{i,j}$ gives
\[ J_{i,j} = \frac{f_{i,j}+\pi/4}{t} = (2g_{i,j}+1)\left(\frac{\pi}{4t}\right) = (2g_{i,j}+1)J \]
for all $1\le i<j\le n$, where we set $J := \pi/(4t) > 0$, whence $t = \pi/(4J)$.  Notice that $J_{i,j}/J = 2g_{i,j}+1$ is an odd integer and
\begin{equation}\label{eqn:J-g-connection}
\frac{J_{i,j}}{J} \equiv_4 \begin{cases}
    1 & \mbox{if $g_{i,j} \equiv_2 0$,} \\
    3 & \mbox{if $g_{i,j} \equiv_2 1$.}
    \end{cases}
\end{equation}

Recall (Lemma~\ref{LemM_n}) that the rows of the matrix $M_n$ are spanned by those indexed by $n$-bit vectors with Hamming weight~$1$.  Letting $S$ be the set of all $x\in\{0,1\}^n$ with $w(x)=1$, it follows that the condition $M_ng \equiv_2 0$ is equivalent to $m_x g \equiv_2 0$ holding for all $x\in S$.
Fix any $r\in [n]$ and let $x\in S$ be such that $x_r=1$ and $x_s=0$ for all $s\ne r$.
Then
\begin{equation} \label{Eulerian}
m_x g \equiv_2 \sum_{1\le i < j\le n} (x_i + x_j)g_{i,j} \equiv_2 \sum_{i<r} g_{i,r} + \sum_{r<j}g_{r,j} \equiv_2 \sum_{i<r\;:\;g_{i,r}\equiv_2 1} g_{i,r} + \sum_{r<j\;:\;g_{r,j}\equiv_2 1} g_{r,j}\;.
\end{equation}

Let $G$ be the graph with vertex set $[n]$ where an edge connects vertices $i<j$ iff $g_{i,j}$ is odd.  Then the right-hand side of \eref{Eulerian} is the degree of the vertex $r$ in $G$.  The condition $m_x g\equiv_2 0$ is then equivalent to the degree of $r$ being even.  Since, $r\in [n]$ (and hence $x\in S$) was chosen arbitrarily, this applies to all the vertices of $G$.  Finally, from \eref{eqn:J-g-connection} we have for all $i<j$ that $J_{i,j}/J \equiv_4 3$ if and only if $g_{i,j}$ is odd, and so the theorem follows.
\end{proof}

Here is an easy restatement of Theorem~\ref{thm:main} that avoids graph concepts.  (Recall that we set $J_{j,i} := J_{i,j}$ for all $i<j$.)

\begin{Cor}
$U_n \propto e^{-iH_nt}$ for some $t>0$ if and only if there exists a constant $J>0$ such that (1.) all $J_{i,j}$ are odd integer multiples of $J$, and (2.) for every $i\in [n]$,
\[ \prod_{j\;:\;j\ne i}\frac{J_{i,j}}{J} \equiv_4 1\;. \]
Furthermore, if $t$ exists, we can set $t := \pi\hbar/4J$.
\end{Cor}

\begin{proof}
Fix $i\in [n]$.  Given that for all $j\ne i$, either $J_{i,j}/J \equiv_4 1$ or $J_{i,j}/J \equiv_4 3$, the product over all such $j$ is congruent to $1 \pmod{4}$ if and only if the latter congruence holds for an even number of such $j$.  This is the stated condition on the graph in Theorem~\ref{thm:main}.
\end{proof}

\section{Parity Versus $U_n$}
\label{sec:U-n-parity}

Fix $n\ge 2$.  Figure~\ref{fig:parity-circuit} gives a quantum circuit $C_n$ implementing the parity gate $P_n$ using a single $U_n$ gate and its inverse $U_n^\dagger$, together with $H$-gates, $S$-gates, and a single $\CNOT$-gate.  In this section we briefly give some related implementations that tighten this result.

First, we observe that $U_n^4 = I$ for all $n$, and $U_n^2 = I$ if $n$ is odd.  Thus $U_n^\dagger$ can be replaced with $U_n^3$ or $U_n$ in the circuit $C_n$, depending on the parity of $n$.  We may also replace the $\CNOT$ gate in $C_n$ with a $U_2$ gate and some $1$-qubit gates.  Letting $\myctrl{Z}$ be the controlled Pauli $z$-gate, we have
\[ \myctrl{Z} = (S^\dagger \otimes S^\dagger)U_2 = U_2(S^\dagger \otimes S^\dagger)\;, \]
which allows us to implement $P_n$ by the following circuit, which is a modification of $C_n$:
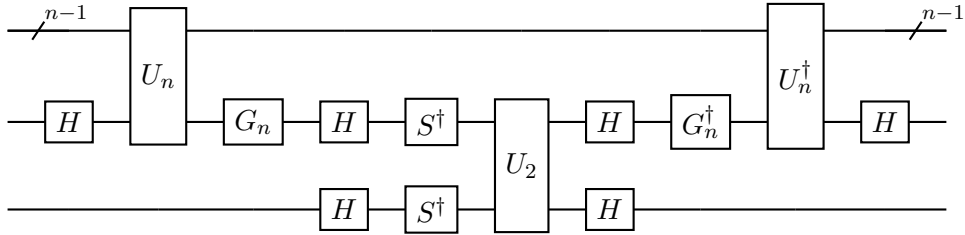
\begin{center}
\begin{quantikz}
& \qw\qwbundle{n-1} & \gate[wires=2]{U_n} &\qw &\qw &\qw &\qw &\qw &\qw & \gate[wires=2]{U_n^\dagger} &\qw & \qwbundle{n-1}\qw & \\
& \gate{H} & & \gate{G_n} & \gate{H} & \gate{S^\dagger}& \gate[wires=2]{U_2} & \gate{H} & \gate{G_n^\dagger} & & \gate{H} & \qw \\
&\qw&\qw&\qw& \gate{H} &\gate{S^\dagger} & &\gate{H} &\qw&\qw&\qw& \qw
\end{quantikz}
\end{center}
Thus $P_n$ can be implemented with at most four $U_n$ gates, a single $U_2$ gate, and constantly many $H$ and $S$ gates.

Conversely, $U_n$ can be implemented with two $P_n$-gates, a few $S$-gates, and an ancilla qubit.  Let $G := S^{2-n}$, which is $Z$, $S$, $I$, or $S^\dagger$, as $n$ is congruent to $0$, $1$, $2$, or $3\pmod{4}$, respectively.  For any $x \in\{0,1\}^n$, one readily checks that
\[ U\ket{x} \otimes \ket{0} = (U_n\otimes I)(\ket{x}\otimes\ket{0}) = P_n (G^{\otimes n}\otimes S) P_n(\ket{x}\otimes\ket{0})\;, \]
where $I$ is the $1$-qubit identity operator.

\section{Couplings Obeying the Inverse Square Law}
\label{sec:inverse-square-law}

Here we consider identical qubits as points in Euclidean space, where the inter-qubit couplings satisfy an inverse square law, i.e., $J_{i,j}$ is proportional to $d_{i,j}^{-2}$, where $d_{i,j}$ is the Euclidean distance between qubits $i$ and $j$.  We find necessary and sufficient conditions for a general arrangement of $n$ such qubits to satisfy the conditions of Theorem~\ref{thm:main}, implementing $U_n$ exactly (up to an overall phase factor).

Our criteria forbid several kinds of arrangements of qubits.  In particular, no three qubits can lie on a line or form a right angle (Theorem~\ref{collinearTheorem}), ruling out many well-studied geometric arrangements of qubits, e.g., meshes.  Our results suggest that an unmodulated inverse square law between identical point qubits is not useful for an exact implementation of $U_n$ for unbounded $n$; either the couplings must be modified in some way or one must make do with an approximation of $U_n$, or both.  See Section~\ref{sec:conclusions} for further discussion.  On the positive side, our criteria allow infinite families of arrangements of four qubits in the plane.

\subsection{General results}

\begin{Def}\label{def:adequacy}
For $n>0$, we say that a set of positive real numbers $\{J_{i,j}\}_{1\le i<j\le n}$ is \emph{adequate for $U_n$} if it satisfies the conditions of Theorem~\ref{thm:main}, i.e., there exists $J>0$ such that all the $J_{i,j}$ are odd integer multiples of $J$, and for each $i\in[n]$, there are an even number of $j\ne i$ such that $J_{i,j}/J \equiv_4 3$.\footnote{We set $J_{j,i} := J_{i,j}$ as before.}  If this is the case, then we call the pairs $\{i,j\}$ such that $J_{i,j}/J \equiv_4 3$ \emph{thick edges}, and the other pairs (where $J_{i,j}/J \equiv_4 1$) \emph{thin edges}.

For $d>0$, we say that a set of $n$ pairwise distinct points $\{p_1,\ldots,p_n\}\subseteq\reals^d$ is \emph{inverse-square adequate for $U_n$} (\emph{isq-adequate} for short) if the set $\{d_{i,j}^{-2}\}_{1\le i<j\le n}$ is adequate for $U_n$, where $d_{i,j}$ is the Euclidean distance between $p_i$ and $p_j$.  We say that $\{p_1,\ldots,p_n\}$ is \emph{weakly isq-adequate} for $U_n$ if the set $\{d_{i,j}^{-2}\}_{1\le i<j\le n}$ satisifies item~(1.) of Theorem~\ref{thm:main} (but not necessarily (2.)), that is, all its elements are odd integer multiples of a common positive $J$.
\end{Def}


Observe that adequacy for $U_n$ only depends on the \emph{ratios} between the $J_{i,j}$; multiplying all the $J_{i,j}$ by the same real constant $c>0$ preserves adequacy for $U_n$: we just scale $J$ by the same constant $c$, and this scaling also preserves thick and thin edges.  It follows that isq-adequacy and weak isq-adequacy for $U_n$ is preserved under any bijective transformations of $\reals^d$ that leave distance ratios invariant, i.e., combinations of rotations, reflections, translations, and dilations by nonzero scaling factors.  We call such transformations \emph{similarities}.\footnote{``\emph{Conformal affine maps}'' may be more descriptive.}

Obviously, isq-adequacy implies weak isq-adequacy.  The reason to consider weak isq-adequacy is to strengthen negative results, which often hold for weak isq-adequacy.  We also have the following:

\begin{Obs}\label{obs:weak-subset}
If a set $S$ of $n$ points is weakly isq-adequate for $U_n$, then any $k$-element subset of $S$ is weakly isq-adequate for $U_k$.
\end{Obs}

Note that Observation~\ref{obs:weak-subset} does not hold for isq-adequacy.

\begin{Thm} \label{collinearTheorem}
If a set of $n \geq 3$ points is weakly isq-adequate for $U_n$, then no three points are collinear nor do they form the vertices of a right triangle.
\end{Thm}

\begin{proof}
Without loss of generality we can assume $n=3$ by Obs.~\ref{obs:weak-subset}.  Suppose $p_1,p_2,p_3$ are three points such that either $d_{1,3} = d_{1,2}+d_{2,3}$ (for collinear points) or $(d_{1,3})^2 = (d_{1,2})^2 + (d_{2,3})^2$ (for points forming right triangle).  Equivalently, $(d_{1,3})^2 = (d_{1,2})^2 + (d_{2,3})^2 + 2bd_{1,2}d_{2,3}$, where $b$ is either $0$ or $1$.
If such an arrangement is weakly isq-adequate for $U_3$, then
there exist real $J>0$ and odd integers $m,n,p$ such that $J_{1,2} = mJ$, $J_{2,3} = nJ$, and $J_{1,3} = pJ$.
For the pairwise couplings, we then have
\begin{align*}
J_{1,2} = mJ&=\frac{1}{(d_{1,2})^2}\;, \\
J_{2,3} = nJ&=\frac{1}{(d_{2,3})^2}\;, \\
J_{1,3} = pJ&=\frac{1}{(d_{1,3})^2}\;.
\end{align*}
From the above equations, we can write
\begin{align} \nonumber
\frac{1}{mJ} + \frac{1}{nJ} - \frac{1}{pJ} &= (d_{1,2})^2 + (d_{2,3})^2 - (d_{1,3})^2 = -2bd_{1,2}d_{2,3} = -\frac{2b}{J\sqrt{mn}} \\ \nonumber
\frac{2b}{J\sqrt{mn}} &= \frac{1}{pJ} - \frac{1}{mJ} - \frac{c}{nJ} \\ \nonumber
\frac{2b}{\sqrt{mn}} &= \frac{mn - np - mp}{pmn} \\
2bp\sqrt{mn} &= mn - np - mp \label{CollinearEqn}
\end{align}
The right-hand side of (\ref{CollinearEqn}) is an odd integer, which requires $b=1$ and $mn$ to be a perfect square, but then the left-hand side is an even integer.  This contradicts the weak isq-adequacy for $U_3$ of any such arrangement.
\end{proof}

Despite the constraints given by Theorem~\ref{collinearTheorem}, there are some nontrivial arrangements that are isq-adequate for $U_n$.  A trivial arrangement is three points forming an equilateral triangle, with all couplings equal (to $J$).  Planar arrangements with four points are harder to come by, but there are some---infinitely many, in fact.  Figure~\ref{fig:four-points} gives two planar arrangements that are adequate for $U_4$.
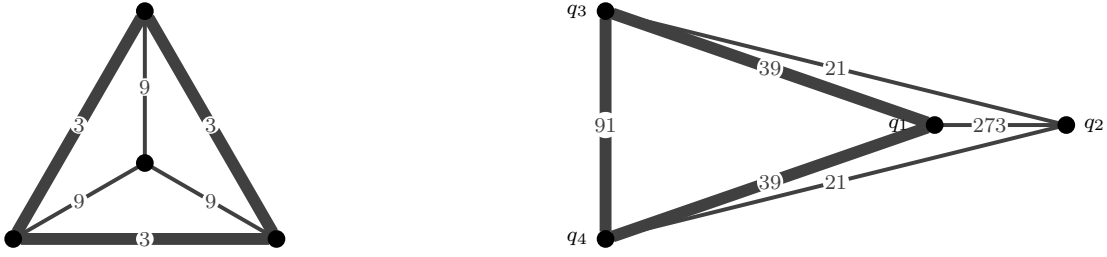
\begin{figure}
    \centering
    \begin{tikzpicture}
    \SetDistanceScale{1.75}
    \SetVertexStyle[FillColor=black,MinSize=0.2\DefaultUnit]
    \Vertex{P1}
    \Vertex[x=1,y=0.57735]{P2}
    \Vertex[x=2]{P3}
    \Vertex[x=1,y=1.732051]{P4}
    \Edge[label=$9$](P1)(P2)
    \Edge[label=$3$,lw=4.5pt](P1)(P3)
    \Edge[label=$3$,lw=4.5pt](P1)(P4)
    \Edge[label=$9$](P2)(P3)
    \Edge[label=$9$](P2)(P4)
    \Edge[label=$3$,lw=4.5pt](P3)(P4)
    \Vertex[label=$q_1$,position=left,distance=1mm,x=7,y=0.8660254]{Q1}
    \Vertex[label=$q_2$,position=right,x=8,y=0.8660254]{Q2}
    \Vertex[label=$q_3$,position=left,x=4.5,y=1.732051]{Q3}
    \Vertex[label=$q_4$,position=left,x=4.5]{Q4}
    \Edge[label=$273$,distance=0.4](Q1)(Q2)
    \Edge[label=$39$,lw=4.5pt](Q1)(Q3)
    \Edge[label=$39$,lw=4.5pt](Q1)(Q4)
    \Edge[label=$21$](Q2)(Q3)
    \Edge[label=$21$](Q2)(Q4)
    \Edge[label=$91$,lw=4.5pt](Q3)(Q4)
    \end{tikzpicture}
\caption{Two possible four-qubit arrangements in the plane that are isq-adequate for $U_4$.  Edges are labeled with the corresponding couplings $J_{i,j}/J$.  Left: an equilateral triangle with a point in the center.  Right: points $q_1,\ldots,q_4$ whose cartesian coordinates can be arranged as: $q_1=(0,0)$, $q_2 = (1,0)$, $q_3=(-5/2,\sqrt 3/2)$, $q_4=(-5/2,-\sqrt 3/2)$.  Thick and thin edges are as shown.}
\label{fig:four-points}
\end{figure}
In each of these, the couplings $kJ$ for $k\equiv_4 3$ form a $3$-cycle.  We will show that this is necessary.

In the rest of this section, we characterize all possible arrangements of four points in $\reals^3$ that are isq-adequate for $U_4$ by giving necessary and sufficient conditions for such arrangements via a form of modular arithmetic on $\rats$.  In Section~\ref{sec:squiggle} we define this form of modular arithmetic and give a few of its basic rules.  Section~\ref{sec:characterization} gives a condition equivalent to isq-adequacy for $U_4$ in terms of these rules.  Finally, in Section~\ref{sec:solve} we use these rules to characterize all possible $3$-dimensional arrangements isq-adequate for $U_4$.  In Section~\ref{sec:U-3}, we briefly characterize isq-adequacy for $U_3$.  In Section~\ref{sec:noU5}, we show that no arrangements in $\reals^3$ are \emph{weakly} isq-adequate for $U_n$ for $n\ge 5$.  It follows that no such arrangements are isq-adequate, either.

\subsection{A modular relation on $\rats$}
\label{sec:squiggle}


Throughout this subsection, we fix an arbitrary prime number $p$.  In the sequel we will only be interested in the case where $p=2$, but it takes no more effort for this section to use an arbitrary $p$.  Recall that for any $x,y\in\reals$ and nonzero $\alpha\in\reals$, we write $x\equiv_\alpha y$ to mean that $(y-x)/\alpha \in\ints$.

\begin{Def}\label{def:squiggle-congruence}
For any $n\in\ints$, let $q := p^n$.  For any $x,y\in\rats$ we write $x\scong q y$ to mean that there exists $k\in\ints$ coprime with $p$ such that $kx \equiv_q ky$.
\end{Def}

Note that if $p=2$, then the condition on $k$ in the above definition is that $k$ be odd.  The next definition is standard.

\begin{Def}\label{def:p-adic-norm}
Every nonzero $r\in\rats$ is uniquely expressible as a product $p^n s$, where $n\in\ints$ and $s\in\rats$ has both numerator and denominator coprime with $p$.  The \emph{$p$-adic norm}
of $r$ is then defined to be $|r|_p := p^{-n}$.
By convention, $|0|_p := 0$.
\end{Def}

It is evident from this definition that $|\mathord{-r}|_p = |r|_p$, that $|rs|_p = |r|_p|s|_p$, and that $|1/r|_p = 1/|r|_p = ||r|_p|_p$, for every $r,s\in\rats$.  We collect basic facts about the $\scong q$ relation in the following lemma:

\begin{Lem}\label{lem:basic-squiggle-congruence}
Let $p$ be prime and let $q := p^n$ for some integer $n\ge 0$.  Let $x,y,z\in\rats$ be any rational numbers.
\begin{enumerate}
\item\label{item:equiv}
$\scong q$ is an equivalence relation.
\item\label{item:equiv-to-approx}
$x\equiv_q y$ implies $x\scong q y$.
\item\label{item:int}
If $x,y\in\ints$, then $x \equiv_q y$ if and only if $x\scong q y$.
\item\label{item:add}
$x \scong q y$ implies $x+z \scong q y+z$.
\item\label{item:scale-up}
$x \scong q y$ if and only if $(\exists s\in\ints)[\; x\scong{pq} y + sq\;]$.
\item\label{item:expand}
For $z\ne 0$, let $m := |z|_p$.  Then $x \scong{qm} y$ iff $xz \scong q yz$.  In particular, $x\scong{q} y$ iff $px \scong{pq} py$.
\item\label{item:recip}
If $x\in\ints$ and $x$ is coprime with $p$, let $x'$ be any integer such that $xx' \equiv_q 1$.  Then $1/x \scong q x'$.
\item\label{item:one-or-three}
If $q>1$, then $|x|_p=1$ iff $(\exists m\in\ints_q^*)[\;x\scong q m\;]$.\footnote{$\ints_q^*$ is the set of integers in the range $0$ through $q-1$ that are coprime with $p$.}  There can be at most one such $m$.
\item\label{item:recip2}
If $x\scong q y$ and $|xy|_p = 1$, then $1/x \scong q 1/y$.
\item\label{item:zero}
$|x|_p \le 1/q$ if and only if $x\scong q 0$.
\end{enumerate}
\end{Lem}

\begin{proof}[Proof sketch]
We sketch proofs of the less obvious facts.  Let $C := \ints\setminus p\ints$ be the set of all integers coprime with $p$.  For transitivity of $\scong q$ in (\ref{item:equiv}.), if $kx \equiv_q ky$ and $\ell y \equiv_q \ell z$ where $k,\ell\in C$, then $k\ell\in C$ and $k\ell x \equiv_q k\ell y \equiv_q k\ell z$.

In (\ref{item:int}.), if $x,y\in\ints$ and $kx\equiv_q ky$ for some $k\in C$, then letting $k'\in C$ be such that $k'k \equiv_q 1$, we have $x \equiv_q k'kx \equiv_q k'ky \equiv_q y$.

For the forward direction of (\ref{item:scale-up}.), if $kx\equiv_q ky$ for some $k\in C$, then $k(x-y) = aq$ for some $a\in\ints$.  Writing $a = mp+r$ for some $m,r\in\ints$, we get $k(x-y) = mpq+rq$, that is, $kx - (ky + rq) = mpq$.  Thus $kx \equiv_{pq} ky+rq \equiv_{pq} k(y+k'rq)$ where $k'\in C$ is such that $kk' \equiv_{pq} 1$.  Set $s:= k'r$.

For (\ref{item:expand}.), let $z = p^r a/b$ for some $r\in\ints$ and $a,b\in C$.  Noting that $m = p^{-r}$, we get
\begin{align*}
x \scong{qm} y &\iff (\exists k\in C)[\;kx \equiv_{qm} ky\;] \iff (\exists k\in C)[\; kp^rx \equiv_q kp^ry\; ] \\
&\iff (\exists k\in C)[\; kap^rx \equiv_q kap^ry\; ] \iff (\exists k\in C)[\; kbzx \equiv_q kbzy\; ] \\
&\iff (\exists k\in C)[\; kzx \equiv_q kzy\; ] \iff zx \scong q zy\;.
\end{align*}
The special case is obtained by setting $z:=1/p$ and using $px$ and $py$ instead of $x$ and $y$.

For (\ref{item:recip}.), $xx' \equiv_q 1 \stackrel{(\ref{item:equiv-to-approx}.)}{\implies} xx' \scong q 1 \stackrel{(\ref{item:expand}.)}{\implies} x' \scong q 1/x$, using $z := 1/x$ for the last implication.

For the forward direction of (\ref{item:one-or-three}.), if $x = a/b$ for $a,b\in C$, then $bx = a \equiv_q bb'a$, where $b'\in C$ satisfies $bb'\equiv_q 1$.  Set $m := b'a\bmod{q}$.  Then $m\in\ints_q^*$ and $bx \equiv_q bm$, yielding $x\scong q m$.  If $m\scong q x \scong q m'$ for $m,m'\in\ints_q$, then $m\equiv_q m'$ by (\ref{item:int}.), whence $m=m'$.  For the reverse direction, if $kx \equiv_q km$ for some $m\in\ints_q^*$ with $k\in C$, then $kx = rq+km$ for some $r\in\ints$, and thus $x = (rq+km)/k$.  Both numerator and denominator are in $C$.

(\ref{item:recip2}.) follows from (\ref{item:expand}.) by setting $z := 1/(xy)$.

For (\ref{item:zero}.), for $k\in C$, \ $kx \equiv_q 0$ iff $(\exists r\in\ints)[\,kx = rq\,]$, iff $(\exists r\in\ints)[\,x = rq/k\,]$, iff $|x|_p\le 1/q$.
\end{proof}

\begin{Def}\label{def:residue}
For $q>1$ a power of prime $p$ and $x\in\rats$ with $|x|_p = 1$, we call the unique $m\in\ints_q^*$ given by Lemma~\ref{lem:basic-squiggle-congruence}(\ref{item:one-or-three}.)\ the \emph{residue} of $x \pmod{q}$.
\end{Def}

\begin{remark}
The key usefulness of (\ref{item:one-or-three}.)\ of Lemma~\ref{lem:basic-squiggle-congruence} is that it allows us to ``pretend'' that an $x\in\rats$ is an integer, provided $|x|_p = 1$.  Fractions with unit $p$-adic norm obey essentially the same rules with respect to $\scong q$ as their residues do with respect to $\equiv_q$.
\end{remark}

Owing to Lemma~\ref{lem:basic-squiggle-congruence}(\ref{item:equiv}.), for any $x\in\rats$ and prime power $q$, we define
\begin{equation}\label{eqn:congruence-class}
[x]_q := \{ y\in\rats : y\scong q x\}\;,
\end{equation}
the equivalence class of $x$ under $\scong q$.  If $q>1$ and $|x|_p = 1$, then the residue of $x\pmod{q}$ is a natural representative element of $[x]_q$.

\begin{Lem}\label{lem:one-multiplier}
Let $p$ be prime, let $n_1,\ldots,n_m$ be arbitrary integers, and let $q_1,\ldots,q_m$ be such that $q_i = p^{n_i}$ for all $i\in[m]$.  Let $x_1,\ldots,x_m,y_1,\ldots,y_m\in\rats$ satisfy $x_i \scong{q_i} y_i$ for all $i\in [m]$.  Then there exists positive $k\in\ints$, coprime with $p$, such that $kx_i \equiv_{q_i} ky_i$ for all $i\in [m]$.
\end{Lem}

\begin{proof}
Let $k_1,\ldots,k_m\in\ints$ be coprime with $p$ such that $k_ix_i \equiv_{q_i} k_iy_i$ for all $i\in [m]$.  Set $k := \lcm(k_1,\ldots,k_m)$.  Then $k$ satisfies the lemma.
\end{proof}


\subsection{A characterization of inverse-square adequacy for $U_4$ in three dimensions}
\label{sec:characterization}

From now on we assume $p=2$.  We let $\oddints := 2\ints+1$ be the set of all odd integers and $\oddrats := \{a/b : a,b\in\oddints\} = \{ x\in\rats : |x|_2 = 1 \}$.

The main theorem of this section is as follows:

\begin{Thm}\label{thm:U-4}
Let $X\subseteq\reals^3$ be a set of four pairwise distinct points.  Then $X$ is isq-adequate for $U_4$ if and only if there exists a similarity $\map{s}{\reals^3}{\reals^3}$ such that
\begin{equation}\label{eqn:standard-position}
s(X) = \{(0,0,0),\,(1,0,0),\,(a/2,b/2,0),\,(c/2,d/2,e/2)\}\;,
\end{equation}
where $a,c,b^2,d^2,bd\in\oddrats$ and there exist $\ell_1,\ell_2\in\{1,3\}$, not both $3$, such that
\begin{align}
a^2+b^2 \scong{16} c^2+d^2+e^2 &\scong{16} 4\ell_1\;,  \label{eqn:abcd-again} \\
a \scong 4 c &\scong 4 \ell_2\;,  \label{eqn:ac3} \\
ac+bd &\scong 8  2(2 - \ell_2)\;. \label{eqn:acbd4}
\end{align}
Supposing this is the case:
\begin{enumerate}
\item\label{item:e}
$e^2\in\rats$ with $e^2 \scong{16} 4(3 - \ell_2)$, that is, $e^2 \scong{16} 8$ if $\ell_2 = 1$ and $e^2 \scong{16} 0$ if $\ell_2 = 3$.
\item
If $\ell_1=\ell_2=1$, then all edges are thin.
\item
If $\ell_1=1$ and $\ell_2=3$, then the thick edges form the $3$-cycle not passing through $(0,0,0)$.
\item
If $\ell_1=3$ and $\ell_2=1$, then the thick edges form the $4$-cycle that excludes the edge $\{(0,0,0),(1,0,0)\}$.
\item
There exist $t,u\in\oddrats$ and positive square-free $n\in\oddints$ with $n\equiv_8 3$ such that $b = t\sqrt n$ and $d = u\sqrt n$.
\end{enumerate}
\end{Thm}

\begin{proof}
Lemma~\ref{lem:basic-squiggle-congruence} is used extensively in this proof.  The numbers with periods in parentheses below refer to the items in that lemma.  The reverse direction of the ``if-and-only-if'' is easier, and we prove it first.  Assume the similarity $s$ is as above with $a,b^2,c,d^2,bd\in\oddrats$ and \erefs{eqn:abcd-again}--(\ref{eqn:acbd4}) satisfied.  We show that $X$ is isq-adequate for $U_4$.

Since similarities preserve isq-adequacy for $U_n$, we can ignore $s$ and assume WLOG that $X = \{p_1,p_2,p_3,p_4\}$, where $p_1 = (0,0,0)$, \ $p_2 = (1,0,0)$, \ $p_3 = (a/2,b/2,0)$, and $p_4 = (c/2,d/2,e/2)$.  For $1\le i<j\le 4$, we let $d_{i,j}$ be the distance between $p_i$ and $p_j$, and we let $J_{i,j} := d_{i,j}^{-2}$ be the corresponding coupling strength.  Thus
\begin{align*}
J_{1,2} &= 1\;, & J_{1,3} &= \frac{4}{a^2+b^2}\;, & J_{1,4} &= \frac{4}{c^2+d^2+e^2}\;, \\
& & J_{2,3} &= \frac{4}{(a-2)^2+b^2}\;, & J_{2,4} &= \frac{4}{(c-2)^2+d^2+e^2}\;, \\
& & & & J_{3,4} &= \frac{4}{(a-c)^2+(b-d)^2+e^2}\;.
\end{align*}
We have $4/J_{1,2} = 4 \scong{16} 4$, and after some manipulation using \eref{eqn:abcd-again} and (\ref{item:add}.), 
\begin{align*}
4/J_{1,3} = a^2+b^2 &\scong{16} 4\ell_1\;, & 4/J_{1,4} = c^2+d^2+e^2 &\scong{16} 4\ell_1\;, \\
4/J_{2,3} = (a-2)^2+b^2 &\scong{16} 4(1+\ell_1-a)\;, & 4/J_{2,4} = (c-2)^2+d^2+e^2 &\scong{16} 4(1+\ell_1-c)\;, \\
& & 4/J_{3,4} = (a-c)^2+(b-d)^2+e^2 &\scong{16} 8 - 2(ac+bd)\;.
\end{align*}
(For the last congruence, note that $8\ell_1 \scong{16} 8$.)  Multiplying everything in (\ref{eqn:ac3}) by $4$ and everything in (\ref{eqn:acbd4}) by $2$ gives $4a\scong{16} 4c\scong{16} 4\ell_2$ and $2(ac+bd) \scong{16} 8 - 4\ell_2$ (by (\ref{item:expand}.)).  Substituting these above, we get
\begin{align*}
4/J_{2,3} &\scong{16} 4(1+\ell_1-\ell_2)\;, & 4/J_{2,4} &\scong{16} 4(1+\ell_1-\ell_2)\;, & 4/J_{3,4}&\scong{16} 4\ell_2\;.
\end{align*}
Dividing everything by $4$ (and using (\ref{item:expand}.)\ again), we obtain
\begin{align*}
1/J_{1,2} &\scong 4 1\;, & 1/J_{1,3} \scong 4 1/J_{1,4} &\scong 4 \ell_1\;, & 1/J_{2,3} \scong 4 1/J_{2,4} &\scong 4 1+\ell_1-\ell_2\;, & 1/J_{3,4} &\scong 4 \ell_2\;.
\end{align*}
Since the right-hand sides are all in $\oddints$, it follows from (\ref{item:expand}.)\ that each $J_{i,j}$ has $2$-adic norm~$1$, i.e., is in $\oddrats$.  Also, $1/3 \scong 4 3$ by (\ref{item:recip}.), and thus by (\ref{item:recip2}.),
\begin{align*}
J_{1,2} &\scong 4 1\;, & J_{1,3} \scong 4 J_{1,4} &\scong 4 \ell_1\;, & J_{2,3} \scong 4 J_{2,4} &\scong 4 1+\ell_1-\ell_2\;, & J_{3,4} &\scong 4 \ell_2\;.
\end{align*}

Finally, by Lemma~\ref{lem:one-multiplier}, there exists a single integer $k>0$ such that
\begin{align*}
kJ_{1,2} &\equiv_4 1\;, & kJ_{1,3} \equiv_4 kJ_{1,4} &\equiv_4 \ell_1\;, & kJ_{2,3} \equiv_4 kJ_{2,4} &\equiv_4 1+\ell_1-\ell_2\;, & kJ_{3,4} &\equiv_4 \ell_2\;.
\end{align*}
Setting $J := 1/k$, we get the following possibilities for the thick edges (i.e., couplings where $J_{i,j}/J \equiv_4 3$), depending on the values of $\ell_1$ and $\ell_2$, and all are evidently isq-adequate for $U_4$:
\begin{itemize}
\item
If $\ell_1 = \ell_2 = 1$, then $J_{i,j}/J\equiv_4 1$ for all $1\le i<j\le 4$, i.e., there are no thick edges.
\item
If $\ell_1 = 1$ and $\ell_2 = 3$, then $J_{2,3}/J \equiv_4 J_{2,4}/J \equiv_4 J_{3,4} \equiv_4 3$ and the rest are $\equiv_4 1$, i.e., the thick edges form a $3$-cycle going through points $p_2,p_3,p_4$.
\item
If $\ell_1 = 3$ and $\ell_2 = 1$, then $J_{1,2}/J \equiv_4 J_{3,4}/J \equiv_4 1$, and the rest are $\equiv_4 3$, i.e., the thick edges form a $4$-cycle that excludes edges $\{p_1,p_2\}$ and $\{p_3,p_4\}$.
\end{itemize}
This proves the reverse implication in Theorem~\ref{thm:U-4}.

\bigskip

For the forward implication, we first suppose only that $\{p_1,p_2,p_3,p_4\}\subset\reals^3$
is \emph{weakly} isq-adequate for $U_4$.  The consequences of weak isq-adequacy we establish here will be used in Section~\ref{sec:noU5}.
Without loss of generality, we can assume that at least one of the edges is thin, say $\{p_1,p_2\}$; otherwise, we replace $J$ with $J/3$, which effectively multiplies all ratios $J_{i,j}/J$ by $3$, flipping the thickness of all edges.
We then let $s$ be a similarity of $\reals^3$ that maps $p_1$ to the origin $(0,0,0)$, $p_2$ to the point $(1,0,0)$ on the $x$-axis, and $p_3$ to a point $(a/2,b/2,0)$ in the $x,y$-plane, for some $a,b\in\reals$.  Such an $s$ clearly exists.  We now let $c,d,e\in\reals$ be such that (letting $s_i := s(p_i)$)
\begin{align*}
s_1 &= (0,0,0)\;, & s_2 &= (1,0,0)\;, & s_3 &= (a/2,b/2,0)\;, & s_4 = (c/2,d/2,e/2)\;.
\end{align*}
Our goal is to show that $a,b^2,c,d^2,bd\in\oddrats$ and $e^2\in\rats$, as well as establishing \erefs{eqn:abcd-again}--(\ref{eqn:acbd4}) and the statements thereafter.

The transformed configuration $\{s_1,s_2,s_3,s_4\}$ is still weakly isq-adequate for $U_4$, with $\{s_1,s_2\}$ still a thin edge.  We must have $b\ne 0$, for otherwise, $s_1,s_2,s_3$ are collinear, violating Theorem~\ref{collinearTheorem}.  For $1\le i<j\le 4$ let $d_{i,j}$ be the distance between $s_i$ and $s_j$.  The couplings $J_{i,j} := d_{i,j}^{-2}$ are thus:
\begin{align}
J_{1,2} &= 1\;, & J_{1,3} &= 4(a^2+b^2)^{-1}\;, & J_{1,4} &= 4(c^2+d^2+e^2)^{-1}\;, \label{eqn:J1} \\
& & J_{2,3} &= 4((a-2)^2+b^2)^{-1}\;, & J_{2,4} &= 4((c-2)^2+d^2+e^2)^{-1}\;, \label{eqn:J2} \\
& & & & J_{3,4} &= 4((a-c)^2+(b-d)^2+e^2)^{-1}\;. \label{eqn:J3}
\end{align}
By the weak isq-adequacy of the $s_i$ for $U_4$, there exists a real $J>0$ such that each ratio $J_{i,j}/J$ is an odd integer; moreover, $1/J = J_{1,2}/J \equiv_4 1$, since $\{s_1,s_2\}$ is a thin edge.  This implies $J\in\rats$, and hence all the $J_{i,j}$ are rational, as are $a^2+b^2$ and $c^2+d^2+e^2$ by \eref{eqn:J1}.  Expanding \eref{eqn:J2} and rearranging, we get
\begin{align*}
\frac{4}{J_{2,3}} &= (a-2)^2+b^2 = a^2+b^2 - 4a + 4\;, & \frac{4}{J_{2,4}} &= (c-2)^2+d^2 = c^2 + d^2 + e^2 - 4c + 4\;,
\end{align*}
whence $a$ and $c$ are both in $\rats$, which then puts $b^2$ and $d^2+e^2$ in $\rats$.  This allows us to use the properties of the $\approx$-congruence given in Lemma~\ref{lem:basic-squiggle-congruence} to reason about $a,b^2,c, d^2+e^2$.

Since $1/J \equiv_4 1$, we have $1/J \scong 4 1$ by (\ref{item:equiv-to-approx}.), and it then follows from (\ref{item:one-or-three}.)\ that $|1/J|_2 = 1 = |J|_2$, that is, $J\in\oddrats$.  From this it further follows that $J\scong 4 1$ by (\ref{item:recip2}.).  For each $i<j$, let $r_{i,j} := (J_{i,j}/J) \bmod{4}$.  Then $r_{i,j}\in\ints_4^* = \{1,3\}$ and $J_{i,j}/J \equiv_4 r_{i,j}$.  For a thin edge, $r_{i,j} = 1$, and for a thick edge, $r_{i,j} = 3$.  We have $r_{1,2} = 1$ by assumption.

By (\ref{item:equiv-to-approx}.), for $i<j$ we have $J_{i,j}/J \scong 4 r_{i,j}$.  Multiplying both sides by $J$, we get $J_{i,j} \scong 4 Jr_{i,j} \scong 4 r_{i,j}$ (applying (\ref{item:expand}.)\ to both $J$ and $r_{i,j}$).  We can thus replace the $J_{i,j}$ with the $r_{i,j}$ in \erefs{eqn:J1}--(\ref{eqn:J3}) above to get the $\scong 4$-congruences
\begin{align*}
r_{1,3} &\scong 4  4(a^2+b^2)^{-1}\;, & r_{1,4} &\scong 4 4(c^2+d^2+e^2)^{-1}\;, \\
r_{2,3} &\scong 4 4((a-2)^2+b^2)^{-1}\;, & r_{2,4} &\scong 4 4((c-2)^2+d^2+e^2)^{-1}\;, \\
r_{3,4} &\scong 4 4((a-c)^2+(b-d)^2+e^2)^{-1}\;.
\end{align*}
All the right-hand sides have $2$-adic norm~$1$ by (\ref{item:one-or-three}.).  Thus by (\ref{item:recip2}.)\ we can take reciprocals of everything to get
\begin{align*}
1/r_{1,3} &\scong 4  (a^2+b^2)/4\;, & 1/r_{1,4} &\scong 4 (c^2+d^2+e^2)/4\;, \\
1/r_{2,3} &\scong 4 ((a-2)^2+b^2)/4\;, & 1/r_{2,4} &\scong 4 ((c-2)^2+d^2+e^2)/4\;, \\
1/r_{3,4} &\scong 4 ((a-c)^2+(b-d)^2+e^2)/4\;.
\end{align*}
Noting that $1/r_{i,j} \scong 4 r_{i,j}$ by (\ref{item:recip}.), we can replace each $1/r_{i,j}$ above by $r_{i,j}$.  Then multiplying everything by $4$ gives (by (\ref{item:expand}.))
\begin{align}
4r_{1,3} &\scong{16}  a^2+b^2\;, & 4r_{1,4} &\scong{16} c^2+d^2+e^2\;, \label{eqn:r1} \\
4r_{2,3} &\scong{16} (a-2)^2+b^2\;, & 4r_{2,4} &\scong{16} (c-2)^2+d^2+e^2\;, \label{eqn:r2} \\
4r_{3,4} &\scong{16} (a-c)^2+(b-d)^2+e^2\;. \label{eqn:r3}
\end{align}

Expanding \erefs{eqn:r2}, then subtracting them from (\ref{eqn:r1}), then dividing everything by $4$ and rearranging, we have
\begin{align}
4r_{2,3} &\scong{16} a^2 + b^2 - 4a + 4 & 4r_{2,4} &\scong{16} c^2 + d^2 + e^2 - 4c + 4\;, \nonumber \\
4(r_{1,3} - r_{2,3}) &\scong{16} 4a - 4 & 4(r_{1,4} - r_{2,4}) &\scong{16} 4c - 4\;, \nonumber \\
a &\scong 4 r_{1,3} - r_{2,3} + 1 &  c &\scong 4 r_{1,4} - r_{2,4} + 1\;. \label{eqn:ac-odd}
\end{align}
The right-hand sides of (\ref{eqn:ac-odd}) are odd integers, so $a,c\in\oddrats$.  This implies $a^2,c^2\in\oddrats$.  Then by (\ref{item:one-or-three}.), $a^2\scong{16} m$ for some $m\in\oddints$, but then \eref{eqn:r1} gives $b^2 \scong{16} 4r_{1,3} - a^2 \scong{16} 4r_{1,3} - m \in\oddints$, whence it follows (from (\ref{item:one-or-three}.)\ again) that $b^2\in\oddrats$.  A similar argument shows that $d^2+e^2\in\oddrats$.

In a similar way, we expand \eref{eqn:r3} and combine it with (\ref{eqn:r1}) to get
\begin{align}
2ac + 2bd &\scong{16} a^2+b^2+c^2+d^2+e^2 - 4r_{3,4} \scong{16} 4(r_{1,3} + r_{1,4} - r_{3,4})\;, \nonumber \\
ac+bd &\scong 8 2(r_{1,3} + r_{1,4} - r_{3,4})\;. \label{eqn:acbd-again}
\end{align}
Thus $bd \scong 8 2(r_{1,3} + r_{1,4} - r_{3,4}) - ac$, which implies $bd\in\oddrats$ because $ac\in\oddrats$.


The following lemma is routine.

\begin{Lem}\label{lem:sq-decomp}
Every nonzero $x\in\rats$ is uniquely expressible in the form $x = t^2n$, where $t\in\rats$ and $n\in\ints$ is positive and square-free.  Also, $x\in\oddrats$ if and only if $t\in\oddrats$ and $n\in\oddints$.
\end{Lem}

\begin{proof}
Let $|x| = p_1^{e_1}\cdots p_k^{e_k}$ be the prime factorization of $|x|$, where $p_1,\ldots,p_k$ are distinct primes and $e_1,\ldots,e_k$ are nonzero integers.  Then we must have
\[ n = \prod_{i:e_i\equiv_2 1} p_i \hspace{.25in}\mbox{ and }\hspace{.25in} t = \pm\sqrt{|x|/n} = \pm p_1^{\floor{e_1/2}} \cdots p_n^{\floor{e_k/2}}\;, \]
and $t$ has the same sign as $x$.  Also, $x\in\oddrats$ if and only if all the $p_i$ are odd, if and only if $n\in\oddints$ and $t\in\oddrats$.
\end{proof}

We can now write $b^2 = t^2n$ for $t\in\oddrats$ with the same sign as $b$ and square-free positive $n\in\oddints$; whence $b = t\sqrt n$.  Since $bd\in\oddrats$, we then have
\[ d = \frac{bd}{b} = \frac{bd}{t\sqrt n} = \frac{bd}{tn}\sqrt n = u\sqrt n\;, \]
where $u := bd/tn$ is evidently in $\oddrats$.  It follows that $d^2 = u^2n\in\oddrats$, and hence $e^2 \in \rats$.  We review an important fact about $\scong q$ where $q$ is a power of~$2$.

\begin{Fact}\label{fact:power-of-2}
Let $q>1$ be a power of $2$.
For any $x,y\in\oddrats$, \ $x^2\scong{2q} y^2 \iff x \scong q \pm y$.  In particular, $x^2 \scong 8 1$ and $x^4 \scong{16} 1$ for any $x\in\oddrats$.
\end{Fact}

Fact~\ref{fact:power-of-2} is easy to verify using (\ref{item:one-or-three}.).  Applying Fact~\ref{fact:power-of-2} to $a$ and $t$ and using \eref{eqn:r1} we get
\begin{align}
a^2 + t^2n \scong{16} 4r_{1,3}
&\scong 8 4 \label{eqn:abcd}
\\
1 + n &\scong 8 4 \nonumber \\
n &\scong 8 3 \nonumber \\
n &\equiv_8 3\;, \label{eqn:n-is-3}
\end{align}
the last line by (\ref{item:int}.).  It follows immediately that
\begin{equation}\label{eqn:u2n}
b^2 \scong{8} d^2 \scong{8} t^2n \scong{8} u^2n \scong{8} 3\;.
\end{equation}
Using this, Fact~\ref{fact:power-of-2}, and \eref{eqn:r1} again, we can solve for $e^2 \pmod{16}$:
\begin{align}
c^2+u^2n+e^2 \scong{16} 4r_{1,4} &\scong{8} 4 \label{eqn:cune} \\
1+3+e^2 &\scong{8} 4 \nonumber \\
e^2 &\scong{8} 0 \nonumber \\
e^2 &\scong{16} (\mbox{$0$ or $8$})\;, \label{eqn:e}
\end{align}
the last line by (\ref{item:scale-up}.).

\begin{remark}
Everything we have established so far---that $a,c,t,u\in\oddrats$, $n\equiv_8 3$, and $e^2\scong 8 0$, as well as \erefs{eqn:r1} through (\ref{eqn:e})---depends only on weak isq-adequacy, with the additional assumption (WLOG) that $r_{1,2} = 1$.
This
will be important in Section~\ref{sec:noU5}, where we consider the inclusion of an additional point.
\end{remark}

From here on we assume that $\{s_1,s_2,s_3,s_4\}$ is isq-adequate, not just weakly so.  From this we observe that $r_{1,3} = r_{1,4}$ and $r_{2,3} = r_{2,4}$: since $\{s_1,s_2\}$ is a thin edge, it must be that the other two edges incident to $s_1$ are either both thick or both thin, i.e., $r_{1,3} = r_{1,4}$; similarly, the two edges connecting $s_2$ to $s_3$ and to $s_4$ are either both thick or both thin, so $r_{2,3} = r_{2,4}$.  Now from (\ref{eqn:ac-odd}) we get $ac \scong 4 (r_{1,3} - r_{2,3} + 1)^2 \scong 4 1$, and this implies by (\ref{item:scale-up}.)\ that either $ac \scong 8 1$ or $ac \scong 8 5$.  We also get from \eref{eqn:acbd-again} that
\begin{equation}\label{eqn:acbd-strong-isq}
ac + bd \scong 8 4 - 2r_{3,4}\;.
\end{equation}
%
%
We now show that the residue of $e^2 \pmod{16}$ depends on $r_{3,4}$ only: $8$ if $r_{3,4} = 1$ and $0$ if $r_{3,4} = 3$.

The thick edges in any isq-adequate configuration for $U_4$ must form either a $3$-cycle, a $4$-cycle, or the empty set; furthermore, we are given that $\{s_1,s_2\}$ is a thin edge ($r_{1,2} = 1$).  We now consider what the status of the edge $\{s_3,s_4\}$ tells us about $e^2$.
%
%
%
%
%
First suppose $\{s_3,s_4\}$ is a thin edge ($r_{3,4} = 1$).  (Note that this rules out a $3$-cycle.)  We show that $e^2 \scong{16} 8$.  Suppose otherwise, i.e.,
$e^2 \scong{16} 0$.
We consider the two subcases $ac\scong 8 1$ and $ac\scong 8 5$ in turn, using Fact~\ref{fact:power-of-2} several times:
\begin{description}
\item[Subcase~1:] $ac\scong 8 1$.  Then multiplying both sides by $c$, we get $a\scong 8 c$, and so $a^2 \scong{16} c^2$.  Then using \erefs{eqn:abcd} and (\ref{eqn:cune}) we get
\begin{align*}
t^2n &\scong{16} u^2n + 0 \\
t^2 &\scong{16} u^2 \\
(tu)^2 = t^2u^2 &\scong{16} u^4 \scong{16} 1 \\
tu &\scong 8 \pm 1 \\
bd = ntu &\scong 8 3tu \scong 8 \pm 3 \\
ac+bd &\scong 8 1 \pm 3 \not\scong 8 2 = 4 - 2r_{3,4}\;,
\end{align*}
which contradicts (\ref{eqn:acbd-strong-isq}).  (Recall that we are currently assuming that $r_{3,4} = 1$.)
\item[Subcase~2:] $ac\scong 8 5$.  Then $a\not\scong 8\pm c$, and so $a^2 \not\scong{16} c^2$.  It follows by \erefs{eqn:abcd} and (\ref{eqn:cune}) that $t^2n \not\scong{16} u^2n$, and dividing both sides by $n$ gives the equivalent $t^2 \not\scong{16} u^2$.  So we have $t \not\scong 8 \pm u$.  Multiplying both sides by $3u$, this is equivalent to $3tu \not\scong 8 \pm 3u^2 \scong 8 \pm 3$, whence we must have $3tu \scong 8 \pm 1$ by (\ref{item:one-or-three}.).  But $3tu \scong 8 ntu = bd$, so $bd\scong 8 \pm 1$, and this gives $ac+bd\scong 8 5\pm 1 \not\scong 8 2 = 4-2r_{3,4}$, which again contradicts (\ref{eqn:acbd-strong-isq}) when $r_{3,4} = 1$.
\end{description}
Thus in either case, we cannot have $e^2 \scong{16} 0$, and so $e^2 \scong{16} 8$ in this case.

Now suppose $\{s_3,s_4\}$ is a thick edge, i.e., $r_{3,4} = 3$.  In this case, there is a $3$-cycle of thick edges involving $s_3$, $s_4$, and either $s_1$ or $s_2$.  If the cycle involves $s_1$, then we can include in the similarity $s$ a reflection through the plane $x=1/2$, which swaps $s_1$ with $s_2$ without affecting $e$.  Thus we can assume without loss of generality that the cycle includes $s_2$, i.e., $r_{1,2} = r_{1,3} = r_{1,4} = 1$ and $r_{2,3} = r_{2,4} = r_{3,4} = 3$.  We show that $e^2 \scong{16} 0$ in this case.  Suppose otherwise, i.e., $e^2 \scong{16} 8$.  This essentially swaps the two subcases for the residue of $ac \pmod{8}$ above:
\begin{itemize}
\item If $ac\scong 8 1$, then $a^2 \scong{16} c^2$, and thus
\begin{align*}
t^2n &\scong{16} u^2n + 8 \\
t^2 &\not\scong{16} u^2 \\
(tu)^2 = t^2u^2 &\not\scong{16} u^4 \scong{16} 1 \\
tu &\not\scong 8 \pm 1 \\
bd = ntu \scong 8 3tu &\not\scong 8 \pm 3 \\
ac+bd &\not\scong 8 1 - 3 \scong 8 6 \scong 8 4 - 2r_{3,4}\;,
\end{align*}
which contradicts (\ref{eqn:acbd-strong-isq}) with $r_{3,4} = 3$.
\item
If $ac \scong 8 5$, then $a^2 \not\scong{16} c^2$, but since $a^2 \scong 8 c^2$, it must be that $c^2 \scong{16} a^2 + 8$ by (\ref{item:scale-up}.).  Then by \eref{eqn:abcd}, we get, as in Case~1 above,
\begin{align*}
a^2 + t^2n &\scong{16} c^2 + u^2n + 8 \scong{16} a^2 + 8 + u^2n + 8 \scong{16} a^2 + u^2n \\
t^2n &\scong{16} u^2n \\
&\;\;\vdots \\
ac+bd \scong 8 5 \pm 3 &\not\scong 8 6 \scong 8 4 - 2r_{3,4}\;,
\end{align*}
which contradicts (\ref{eqn:acbd-strong-isq}) with $r_{3,4} = 3$.
\end{itemize}
Thus in either subcase, $e^2\scong{16} 0$.

To summarize, we have shown that configurations in $\reals^3$ isq-adequate for $U_4$ fall into three types according to which edges are thick (i.e., which $r_{i,j} = 3$), each type corresponding to a choice of $(\ell_1,\ell_2)$ to match up with the statement of Theorem~\ref{thm:U-4}:
\begin{description}
\item[No thick edges:] Then $e^2 \scong{16} 8$, and \erefs{eqn:r1}, (\ref{eqn:ac-odd}), and (\ref{eqn:acbd-again}) give
\begin{align}
a^2+b^2 = a^2 + t^2n &\scong{16} 4 & c^2+d^2 = c^2 + u^2n &\scong{16} 12 \label{eqn:all-thin-11-1} \\
a\scong 4 c&\scong 4 1 & ac+bd = ac+tun &\scong 8 2 \label{eqn:all-thin-11-2}
\end{align}
Thus $(\ell_1,\ell_2) = (1,1)$.
\item[$3$-cycle avoiding $(0,0,0)$:] Then $e^2 \scong{16} 0$, and \erefs{eqn:r1}, (\ref{eqn:ac-odd}), and (\ref{eqn:acbd-again}) give
\begin{align}
a^2+b^2 = a^2 + t^2n &\scong{16} 4 & c^2+d^2 = c^2 + u^2n &\scong{16} 4 \label{eqn:3-cycle-13-1} \\
a\scong 4 c &\scong 4 3 & ac+bd = ac+tun &\scong 8 6 \label{eqn:3-cycle-13-2}
\end{align}
Thus $(\ell_1,\ell_2) = (1,3)$.  All planar configurations are of this type.
\item[$4$-cycle without the edge $\{(0,0,0),(1,0,0)\}$:] Then $e^2 \scong{16} 8$, and \erefs{eqn:r1}, (\ref{eqn:ac-odd}), and (\ref{eqn:acbd-again}) give
\begin{align}
a^2+b^2 = a^2 + t^2n &\scong{16} 12 & c^2+d^2 = c^2 + u^2n &\scong{16} 4 \label{eqn:4-cycle-31-1} \\
a\scong 4 c&\scong 4 1 & ac+bd = ac+tun &\scong 8 2 \label{eqn:4-cycle-31-2}
\end{align}
Thus $(\ell_1,\ell_2) = (3,1)$.
\end{description}
Here, $b = t\sqrt n$ and $d = u\sqrt n$ for $t,u\in\oddrats$ and positive square-free $n\equiv_8 3$.  

This concludes the proof of Theorem~\ref{thm:U-4}.
\end{proof}

The next corollary will be helpful in the next section to classify the solutions to \erefs{eqn:abcd-again}--(\ref{eqn:acbd4}).

\begin{Cor}\label{cor:mod-8}
For the configuration $s(X)$ given in \eref{eqn:standard-position} of Theorem~\ref{thm:U-4}, assuming $b = t\sqrt n$ and $d = u\sqrt n$ where $t,u\in\oddrats$ and $n\in\oddints$ is square-free and $n\equiv_8 3$, the isq-adequacy of $s(X)$ for $U_4$ depends \emph{only} on the residues of $a,c,t,u \pmod{8}$ and the residues of $e^2$ and $n\pmod{16}$.  That is, substituting any value $\scong{8}$ to the value of $a$, $c$, $t$, or $u$ above or substituting any square-free $n'\equiv_{16} n$ for $n$ or $(e')^2\scong{16} e^2$ for $e^2$ above preserves isq-adequacy for $U_4$.
\end{Cor}

\begin{proof}
Any such substitution leaves \erefs{eqn:abcd-again}--(\ref{eqn:acbd4}) invariant (or equivalently, \erefs{eqn:all-thin-11-1}--(\ref{eqn:4-cycle-31-2})).
\end{proof}


\subsection{Classifying configurations in $\reals^3$ isq-adequate for $U_4$}
\label{sec:solve}

In this section we apply Theorem~\ref{thm:U-4}
to classify all three-dimensional arrangements of four identical qubits that are isq-adequate for $U_4$.  Up to similarities and permutations of the qubits, these arrangements fall into well-defined groups, according to the integer residues in $\ints_{16}^*$ of the coordinates.

Suppose we have four points $\{s_1,s_2,s_3,s_4\}$ isq-adequate for $U_4$, with $a,b,c,d,e,t,u,n$ as in Theorem~\ref{thm:U-4}.  Plugging $b = t\sqrt n$ and $d = u\sqrt n$ into \erefs{eqn:abcd-again}--(\ref{eqn:acbd4}), we get the equivalent congruences
\begin{align}
a^2+t^2n \scong{16} c^2+u^2n+e^2 &\scong{16} 4\ell_1\;,  \label{eqn:atncun} \\
a \scong 4 c &\scong 4 \ell_2\;,  \label{eqn:ac3-again} \\
ac+tun &\scong 8  2(2-\ell_2)\;, \label{eqn:actun4} \\
e^2 &\scong{16} 4(3-\ell_2)\;, \label{eqn:e16}
\end{align}
where $(\ell_1,\ell_2) \in \{(1,1),(1,3),(3,1)\}$.  Given fixed $n$, we wish to find a succinct classification of all solutions to \erefs{eqn:atncun}--(\ref{eqn:e16}).  By Corollary~\ref{cor:mod-8} we can restrict our attention to the special case where $a,c,t,u\in\ints_8^* = \{1,3,5,7\}$ and $n\equiv_8 3$ (positive and square-free).
Any other solution to \erefs{eqn:atncun}--(\ref{eqn:e16}) is then obtained by freely substituting any element of $[a]_8$ in for $a$ and likewise for $c$, $t$, and $u$ independently, and similarly for any element of $[e^2]_{16}$ for $e^2$.  To summarize, given positive square-free $n\equiv_8 3$, every solution to \erefs{eqn:atncun}--(\ref{eqn:e16}) corresponds uniquely to a solution where $a,c,t,u\in\{1,3,5,7\}$, the correspondence obtained by the substitutions described above.  We call these latter solutions \emph{representative solutions}, of which there are clearly finitely many, given fixed $n$ and $e$.

Note that substituting any positive square-free $n'\equiv_{16} n$ for $n$ also leaves \erefs{eqn:atncun}--(\ref{eqn:e16}) invariant, thus giving the same representative solutions for $a,c,t,u,e^2$.

For the rest of this section, we assume that $a,c,t,u \in \{1,3,5,7\}$.  Restricted to these sets, $\scong 8$ coincides with equality for $a,c,t,u$.  Further, $\scong q$ can be changed to $\equiv_q$ in \erefs{eqn:atncun}--(\ref{eqn:e16}) for $q\in\{4,8,16\}$.  From \eref{eqn:ac3-again}, we get $a\equiv_4 c$, whence $ac \equiv_8 1$ if $a = c$, and $ac \equiv_8 5$ otherwise.

\subsubsection{The $e^2 \scong{16} 0$ case}
\label{sec:e-scong-0}

We first consider the case where $e^2 \scong{16} 0$, which includes all planar configurations.  By Theorem~\ref{thm:U-4}, the thick edges form a $3$-cycle passing through points $s_2,s_3,s_4$, and $(\ell_1,\ell_2) = (1,3)$.  (See \erefs{eqn:3-cycle-13-1} and (\ref{eqn:3-cycle-13-2}).)  From \eref{eqn:ac3-again}, we get $a,c\in\{3,7\}$.  Since $n\equiv_8 3$, \eref{eqn:actun4} becomes
\begin{align}
ac + 3tu &\equiv_8 6 \nonumber\\
3tu &\equiv_8 6 - ac \nonumber\\
tu &\equiv_8 3(6-ac) \equiv_8 2-3ac \nonumber\\
&\equiv_8 (\mbox{$7$ if $a = c$ and $3$ otherwise}) \nonumber\\
u &\equiv_8 (\mbox{$7t$ if $a = c$ and $3t$ otherwise}) \equiv_8 (\mbox{$-t$ if $a = c$ and $3t$ otherwise}) \label{eqn:tu}
\end{align}

We have $n \equiv_{16} (\mbox{$3$ or $11$})$.  We now consider the case where $n\equiv_{16} 3$.  The case where $n\equiv_{16} 11$ will be handled similarly.  \eref{eqn:atncun} becomes
\begin{align}
a^2 + 3t^2 &\equiv_{16} 4 & c^2 + 3u^2 &\equiv_{16} 4 \nonumber\\
t^2 &\equiv_{16} a^2 & u^2 &\equiv_{16} c^2 \nonumber\\
t &\equiv_8 \pm a\;, & u &\equiv_8 \pm c\;. \label{eqn:tu-by-ac-3}
\end{align}
The second line follows from the fact that each squared value is either $1$ or $9 \pmod{16}$.  The last line uses Fact~\ref{fact:power-of-2}.  Note that these steps are reversible; implications run in both directions.

Combining the constraints given by \erefs{eqn:tu} and (\ref{eqn:tu-by-ac-3}), we have the following possible representative solutions for $a$, $c$, $t$, and $u$ when $n\equiv_{16} 3$:
\[ \begin{array}{|c||c|c|c|c|c|c|c|c|} \hline
  & A & B & C & D & E & F & G & H \\\hline\hline
a & 3 & 3 & 3 & 3 & 7 & 7 & 7 & 7 \\\hline
c & 3 & 3 & 7 & 7 & 3 & 3 & 7 & 7 \\\hline
t & 3 & 5 & 3 & 5 & 1 & 7 & 1 & 7 \\\hline
u & 5 & 3 & 1 & 7 & 3 & 5 & 7 & 1 \\\hline
\end{array} \]
It is readily checked that each solution $A$--$H$ satisfies \erefs{eqn:atncun}--(\ref{eqn:actun4}) when $n\equiv_{16} 3$.

For $n\equiv_{16} 11$, \eref{eqn:tu} still holds, and a calculation similar to that for \eref{eqn:tu-by-ac-3} yields
\begin{align}
t^2 &\equiv_{16} a^2 + 8 & u^2 &\equiv_{16} c^2 + 8 \nonumber\\
t &\equiv_8 4\pm a & u &\equiv_8 4\pm c \label{eqn:tu-by-ac-11}
\end{align}
Using \eref{eqn:tu-by-ac-11}, we append the solutions for $n\equiv_{16} 11$ to the table above to get a complete table of $16$ representative solutions for $a,c,t,u$, split into two groups depending on $[n]_{16}$:
\begin{equation}\label{eqn:solutions-table}
 \begin{array}{|c||c|c|c|c|c|c|c|c||c|c|c|c|c|c|c|c|} \hline
  & \multicolumn{8}{c||}{n\equiv_{16} 3} & \multicolumn{8}{c|}{n\equiv_{16} 11} \\\hline
  & A & B & C & D & E & F & G & H & I & J & K & L & M & N & O & P \\\hline\hline
a & \phantom{.}3\phantom{.} & \phantom{.}3\phantom{.} & \phantom{.}3\phantom{.} & \phantom{.}3\phantom{.} & \phantom{.}7\phantom{.} & \phantom{.}7\phantom{.} & \phantom{.}7\phantom{.} & \phantom{.}7\phantom{.} & \phantom{.}3\phantom{.} & \phantom{.}3\phantom{.} & \phantom{.}3\phantom{.} & \phantom{.}3\phantom{.} & \phantom{.}7\phantom{.} & \phantom{.}7\phantom{.} & \phantom{.}7\phantom{.} & \phantom{.}7\phantom{.} \\\hline
c & 3 & 3 & 7 & 7 & 3 & 3 & 7 & 7 & 3 & 3 & 7 & 7 & 3 & 3 & 7 & 7 \\\hline
t & 3 & 5 & 3 & 5 & 1 & 7 & 1 & 7 & 1 & 7 & 1 & 7 & 3 & 5 & 3 & 5 \\\hline
u & 5 & 3 & 1 & 7 & 3 & 5 & 7 & 1 & 7 & 1 & 3 & 5 & 1 & 7 & 5 & 3 \\\hline
\end{array}
\end{equation}

Applying this classification yields the following:

\begin{Prp}\label{prop:no-trapezoid}
The corners of a trapezoid are never isq-adequate for $U_4$, that is, no configuration isq-adequate for $U_4$ can have two parallel edges.
\end{Prp}

\begin{proof}
Since a trapezoid is planar, if its four corners are isq-adequate for $U_4$, they satisfy the $e\scong{16} 0$ case, whence the thick edges form a $3$-cycle.  It follows that one of the two parallel edges is thin, and so the points can be mapped by a similarity to
\[ \{(0,0),(1,0),(a,t\sqrt n),(c,u\sqrt n)\}\;, \]
where $\{(0,0),(1,0)\}$ is one of the parallel edges, and where $a,c,t,u$ have residue classes (mod~$8$) given by one of the columns in \eref{eqn:solutions-table}.  As the edge $\{(a,t\sqrt n),(c,u\sqrt n)\}$ is parallel to $\{(0,0),(1,0)\}$, we have $t = u$, but there is no column of \eref{eqn:solutions-table} where $t$ and $u$ have the same residue, making this configuration impossible.
\end{proof}

The solutions given in (\ref{eqn:solutions-table}) are not all qualitatively distinct, due to some geometrical symmetries that preserve edge thickness.  An obvious one is to reflect in the $xy$-plane, negating $e$ but leaving all other values unaltered.  This leaves all the solution types $A$--$P$ unchanged.  If we reflect the points in the $xz$-plane, however, we negate $b$ and $d$ (and $t$ and $u$) simultaneously.
This leaves \erefs{eqn:abcd-again}--(\ref{eqn:acbd4}) invariant and allows us to group solutions $A$--$P$ into pairs of equivalent solutions: $A\leftrightarrow B$, $C\leftrightarrow D$, $E\leftrightarrow F$, $G\leftrightarrow H$, $I\leftrightarrow J$, $K\leftrightarrow L$, $M\leftrightarrow N$, $O\leftrightarrow P$, as shown below:
%
\begin{equation}\label{eqn:solutions-table-grouped}
 \begin{array}{|c||cc|cc|cc|cc||cc|cc|cc|cc|} \hline
  & \multicolumn{8}{c||}{n\equiv_{16} 3} & \multicolumn{8}{c|}{n\equiv_{16} 11} \\\hline
  & \multicolumn{2}{c|}{AB} & \multicolumn{2}{c|}{CD} & \multicolumn{2}{c|}{EF} & \multicolumn{2}{c||}{GH} & \multicolumn{2}{c|}{IJ} & \multicolumn{2}{c|}{KL} & \multicolumn{2}{c|}{MN} & \multicolumn{2}{c|}{OP} \\\hline\hline
a & 3 & 3 & 3 & 3 & 7 & 7 & 7 & 7 & 3 & 3 & 3 & 3 & 7 & 7 & 7 & 7 \\\hline
c & 3 & 3 & 7 & 7 & 3 & 3 & 7 & 7 & 3 & 3 & 7 & 7 & 3 & 3 & 7 & 7 \\\hline
t & 3 & 5 & 3 & 5 & 1 & 7 & 1 & 7 & 1 & 7 & 1 & 7 & 3 & 5 & 3 & 5 \\\hline
u & 5 & 3 & 1 & 7 & 3 & 5 & 7 & 1 & 7 & 1 & 3 & 5 & 1 & 7 & 5 & 3 \\\hline
\end{array}
\end{equation}

A typical example solution (planar, of type $GH$) has
\begin{align*}
s_1 &= (0,0,0)\;, & s_2 &= (1,0,0)\;, & s_3 &= (-1/2,\sqrt 3/2,0)\;, & s_4 &= (-1/2,-\sqrt 3/2,0)\;.
\end{align*}
The thick edges form an equilateral triangle with side length $\sqrt 3$ in the $xy$-plane centered at the origin, shown rotated on the left in Figure~\ref{fig:four-points}.  (The right configuration is also of type $GH$.)  This solution is particularly nice; it has the smallest possible upper bound on the ratios $J_{i,j}/J$ of any planar solution: $9$ for the thin edges, $3$ for the thick ones.

Another possible symmetry is to rotate about the $x$-axis to bring $s_4$ into the $xy$-plane, then swap the roles of $s_3$ and $s_4$.  This symmetry fixes $s_1$ and $s_2$ while mapping $s_3 = (a,b,0)/2 \mapsto (c',d',e')/2$ and $s_4 = (c,d,e)/2 \mapsto (a',b',0)/2$, where (letting $r := \sqrt{d^2+e^2}$)
\begin{align*}
a' &= c & b' &= r & c' &= a & d' &= bd/r & e' &= -be/r
\end{align*}
(There are actually two rotations that can accomplish this, differing by $180^\circ$.  We have arbitrarily chosen the one that makes $b' > 0$.)  This transformation leaves \erefs{eqn:abcd-again}--(\ref{eqn:acbd4}) invariant, but may change the residue of $n \pmod{16}$, due to $n$ having to stay square-free.  For example, if $d = \sqrt 3$ (so $n=3$) and $e = 4\sqrt 6$, then $(b')^2 = r^2 = 99$, which is not square-free; we must write $b' = 3\sqrt{11}$ instead, making the value of $n$ change to $11$.  On the other hand, for planar configurations, $e = 0$ and no rotation is needed; we merely swap $s_3$ with $s_4$.  In this case the required value of $n$ stays the same, and we can coalesce pairs $CD\leftrightarrow EF$ and $KL\leftrightarrow MN$ into single groups.

\subsubsection{Other symmetries}

This subsection may be skipped as nothing else in the paper depends on it.

There are other possible similarities that preserve edge thickness.  A unique rotation about the $z$-axis followed by a unique dilation maps the point $s_3 = (a,t\sqrt n,0)/2$ to the point $s_2 = (1,0,0)$.  This similarity, it can be shown, preserves the required value of $n$, fixes the origin $s_1$, maps $s_2$ to $(a',t'\sqrt n,0)/2$ and $s_4 = (c,u\sqrt n,e)/2$ to $(c',u'\sqrt n,e)/2$, where (letting $\gamma := 4/(a^2+b^2) = 4/(a^2+t^2n)$)
\begin{align}
a' &= \gamma a\;, & c' &= \gamma\left(\frac{ac + tun}{2}\right)\;, & t' &= -\gamma t\;, & u' &= \gamma\left(\frac{au - ct}{2}\right)\;. \label{eqn:gamma}
\end{align}
By \eref{eqn:atncun} and Lemma~\ref{lem:basic-squiggle-congruence}(\ref{item:expand}.), $\gamma \scong 4 1/\gamma \scong 4 1$.  Likewise, $a',c',t',u'$ are only determined up to $\scong 4$-congruence if $a,c,t,u$ are only known up to $\scong 8$-congruence.  This is due to the divisions by $2$ (or by $4$ in the case of $\gamma$) above.  For the same reason, determining $a'c',t',u'$ up to $\scong 8$-congruence depends on what $a,c,t,u$ are up to $\scong{16}$-congruence.  Thus values of $a,c,t,u$ with the same representative solution may map to different representative solutions under this similarity.

From \eref{eqn:gamma} above, one can see that adjusting $a$, say, by adding $\pm 8$ to it will adjust the residue (mod~$8$) of $\gamma$, $a'$, and $t'$ by $\pm 4$, leaving the residues (mod~$8$) of $c'$ and $u'$ unchanged, leading to a different representative solution.  This is even though adjusting $a$ by $\pm 8$ does not change its residue (mod~$8$).  Adjusting $t$ by $\pm 8$ causes the exact same adjustments to these residues as the $a$-adjustment.  Similarly, adjusting either $c$ or $u$ by $\pm 8$ causes an adjustment of $\pm 4$ of the residues $c'$ and $u'$ (mod~$8$) while leaving those of $\gamma$, $a'$, and $t'$ unchanged.  Depending on the residues of $a,c,t,u$ (mod~$16$), one can show that:
\begin{itemize}
\item
For $n\equiv_{16} 3$, any representative solution in the set $\{A,C,F,H\}$ for $a,c,t,u$ may map to any representative solution in the set $\{B,D,E,G\}$ for $a',c',t',u'$ and vice versa.
\item
For $n\equiv_{16} 11$, any representative solution in the set $\{I,K,N,P\}$ for $a,c,t,u$ may map to any representative solution in the set $\{J,L,M,O\}$ for $a',c',t',u'$ and vice versa.
\end{itemize}
These are the only possible mappings, and they were found with the aid of a computer.

\subsubsection{The $e^2\scong{16} 8$ cases}
\label{sec:e-scong-8}

Here we forgo detailed derivations (which are similar to the previous section) and only give the key results along with two examples.  There are two possible configurations of thick edges---either none or a $4$-cycle.  In either case, $\ell_2 = 1$ and $\{s_1,s_2\}$ and $\{s_3,s_4\}$ are thin.  We have $a,c\in\{1,5\}$ and $ac+3tu \scong 8 2$ in both cases by \erefs{eqn:all-thin-11-2} and (\ref{eqn:4-cycle-31-2}), which implies (cf.~\eref{eqn:tu})
\begin{equation}\label{eqn:tu-again}
u  \equiv_8 (\mbox{$3t$ if $a = c$ and $-t$ otherwise})
\end{equation}

If there are no thick edges, then \eref{eqn:all-thin-11-1} implies
\begin{align*}
t &\equiv_8 \begin{cases}
\pm a & \mbox{if $n\equiv_{16} 3$,} \\
4 \pm a & \mbox{if $n\equiv_{16} 11$,}
\end{cases}
&
u &\equiv_8 \begin{cases}
4\pm c & \mbox{if $n\equiv_{16} 3$,} \\
\pm c  & \mbox{if $n\equiv_{16} 11$.}
\end{cases}
\end{align*}
The table below gives the representative solutions when there are no thick edges.  Solutions are grouped together that represent reflections of each other in the $xz$-plane.
\begin{equation}\label{eqn:solutions-all-thin-table-grouped}
 \begin{array}{|c||cc|cc|cc|cc||cc|cc|cc|cc|} \hline
  & \multicolumn{8}{c||}{n\equiv_{16} 3} & \multicolumn{8}{c|}{n\equiv_{16} 11} \\\hline
  & \multicolumn{2}{c|}{AB} & \multicolumn{2}{c|}{CD} & \multicolumn{2}{c|}{EF} & \multicolumn{2}{c||}{GH} & \multicolumn{2}{c|}{IJ} & \multicolumn{2}{c|}{KL} & \multicolumn{2}{c|}{MN} & \multicolumn{2}{c|}{OP} \\\hline\hline
a & 1 & 1 & 1 & 1 & 5 & 5 & 5 & 5 & 1 & 1 & 1 & 1 & 5 & 5 & 5 & 5 \\\hline
c & 1 & 1 & 5 & 5 & 1 & 1 & 5 & 5 & 1 & 1 & 5 & 5 & 1 & 1 & 5 & 5 \\\hline
t & 1 & 7 & 1 & 7 & 3 & 5 & 3 & 5 & 3 & 5 & 3 & 5 & 1 & 7 & 1 & 7 \\\hline
u & 3 & 5 & 7 & 1 & 5 & 3 & 1 & 7 & 1 & 7 & 5 & 3 & 7 & 1 & 3 & 5 \\\hline
\end{array}
\end{equation}
The best example of an all-thin-edge solution
is the regular tetrahedron, which is of type $AB$:
\begin{align*}
s_1 &= (0,0,0)\;, & s_2 &= (1,0,0)\;, & s_3 &= \left(\frac{1}{2},\frac{\sqrt 3}{2},0\right) & s_4 &= \left(\frac{1}{2},\frac{\sqrt 3}{6},\frac{\sqrt 6}{3}\right)
\end{align*}
All edges have unit length, so we take $J := 1$.

\erefs{eqn:all-thin-11-1} and (\ref{eqn:4-cycle-31-1}) are symmetric in the sense that we get one from the other by swapping points $s_3$ with $s_4$, i.e., $(a,t)$ with $(c,u)$.  Also, each of \eref{eqn:all-thin-11-2} and \eref{eqn:4-cycle-31-2} is by itself invariant under the same swap.  Thus the representative solutions for the thick edges forming a $4$-cycle are obtained from \eref{eqn:solutions-all-thin-table-grouped} via this swap:
\begin{equation}\label{eqn:solutions-4-cycle-table-grouped}
 \begin{array}{|c||cc|cc|cc|cc||cc|cc|cc|cc|} \hline
  & \multicolumn{8}{c||}{n\equiv_{16} 3} & \multicolumn{8}{c|}{n\equiv_{16} 11} \\\hline
  & \multicolumn{2}{c|}{AB} & \multicolumn{2}{c|}{CD} & \multicolumn{2}{c|}{EF} & \multicolumn{2}{c||}{GH} & \multicolumn{2}{c|}{IJ} & \multicolumn{2}{c|}{KL} & \multicolumn{2}{c|}{MN} & \multicolumn{2}{c|}{OP} \\\hline\hline
a & 1 & 1 & 1 & 1 & 5 & 5 & 5 & 5 & 1 & 1 & 1 & 1 & 5 & 5 & 5 & 5 \\\hline
c & 1 & 1 & 5 & 5 & 1 & 1 & 5 & 5 & 1 & 1 & 5 & 5 & 1 & 1 & 5 & 5 \\\hline
t & 3 & 5 & 5 & 3 & 7 & 1 & 1 & 7 & 1 & 7 & 7 & 1 & 5 & 3 & 3 & 5 \\\hline
u & 1 & 7 & 3 & 5 & 1 & 7 & 3 & 5 & 3 & 5 & 1 & 7 & 3 & 5 & 1 & 7 \\\hline
\end{array}
\end{equation}
(The columns have been rearranged to have the same $(a,c)$ order.\footnote{It is interesting to note that (\ref{eqn:solutions-all-thin-table-grouped}) and (\ref{eqn:solutions-4-cycle-table-grouped}) are obtained from each other by swapping the residues of $n\pmod{16}$.})  A good example of a $4$-cycle solution to \erefs{eqn:4-cycle-31-1} and (\ref{eqn:4-cycle-31-2}) is an elongated tetrahedron, this one being of type $IJ$:
\begin{align*}
s_1 &= (0,0,0)\;, & s_2 &= (1,0,0)\;, & s_3 &= \left(\frac{1}{2},\frac{\sqrt{11}}{2},0\right) & s_4 &= \left(\frac{1}{2},\frac{9\sqrt{11}}{22},\frac{\sqrt{110}}{11}\right)
\end{align*}
The edges $\{s_1,s_2\}$ and $\{s_3,s_4\}$ have unit length, and the rest have lenth $\sqrt 3$.  We take $J := 1/9$.


\section{Isq-Adequacy for $U_3$ of Configurations in $\reals^2$}
\label{sec:U-3}

In this section we give characterizations of $3$-point configurations isq-adequate for $U_3$.  These characterizations are obtained using techniques similar to those used in Section~\ref{sec:e-scong-0}, and we only state the results here.
The following theorem is analogous to Theorem~\ref{thm:U-4} and has a similar proof, which we omit.

\begin{Thm}\label{thm:U-3}
Let $X\subseteq\reals^2$ be a set of three pairwise distinct points.  Then $X$ is isq-adequate for $U_3$ if and only if there exists a similarity $\map{s}{\reals^2}{\reals^2}$ such that $s(X) = \{(0,0), (1,0), (a/2,b/2)\}$, where $a,b^2\in\oddrats$ and
\begin{align}
a^2+b^2 &\scong{16} 4\;, & a &\scong 4 1\;. \label{eqn:U-3}
\end{align}
Supposing this is the case, all edges are thin, and there exist $t\in\oddrats$ and positive square-free $n\in\oddints$ with $n\equiv_8 3$ such that $b = t\sqrt n$.
\end{Thm}

As with $U_4$, isq-adequacy for $U_3$ only depends on $[a]_8$, $[t]_8$, and $[n]_{16}$.  The possible representative solutions (mod~$8$) to \eref{eqn:U-3} for the two possibilities for $[n]_{16}$ are given by the following table.  Solutions that are symmetric under reflection in the $x$-axis (i.e., swapping $t$ with $-t$) are grouped together.
\begin{equation}\label{eqn:U3-solutions-table-grouped}
 \begin{array}{|c||cc|cc||cc|cc|} \hline
  & \multicolumn{4}{c||}{n\equiv_{16} 3} & \multicolumn{4}{c|}{n\equiv_{16} 11} \\\hline
  & \multicolumn{2}{c|}{AB} & \multicolumn{2}{c||}{CD} & \multicolumn{2}{c|}{EF} & \multicolumn{2}{c|}{GH} \\\hline\hline
a & 1 & 1 & 5 & 5 & 1 & 1 & 5 & 5 \\\hline
t & 1 & 7 & 3 & 5 & 3 & 5 & 1 & 7 \\\hline
\end{array}
\end{equation}

\section{No $5$-Point Configuration in $\reals^3$ Is Weakly Isq-Adequate for $U_5$}
\label{sec:noU5}

In this section we prove the following theorem:

\begin{Thm}\label{thm:noU5}
No set of five points in $\reals^3$ is weakly isq-adequate for $U_5$.
\end{Thm}

\begin{proof}
Suppose $\{p_1,\ldots,p_5\}\subseteq\reals^3$ is weakly isq-adequate for $U_5$, with couplings $\{J_{ij} : 1\le i<j\le 5\}$ such that each ratio $J_{ij}/J\in\ints$ is odd, for some fixed $J>0$.  As in the proof of the forward implication of Theorem~\ref{thm:U-4}, there is a similarity that maps $p_1,\ldots,p_4$ to the points $s_1 = (0,0,0)$, $s_2 = (1,0,0)$, $s_3 = (a/2,b/2,0)$, and $s_4 = (c/2,d/2,e/2)$ for real numbers $a,b,c,d,e$, where, without loss of generality, $\{s_1,s_2\}$ is a thin edge.  This map sends $p_5$ to $s_5 := (f/2,g/2,h/2)$ for some $f,g,h\in\reals$.


The congruence relations (\ref{eqn:r1})--(\ref{eqn:e}) for $s_1,\ldots,s_4$ obtained in the proof of Theorem~\ref{thm:U-4} only depend on weak isq-adequacy and so are still valid here, and we still have $b = t\sqrt n$ and $d = u\sqrt n$ where $a,c,t,u\in\oddrats$ and $n\in\oddints$ is square-free and $n\equiv_8 3$.  (See the Remark following \eref{eqn:e}.)  For convenience, we reproduce the congruences we need here:
\begin{align}
c^2+d^2+e^2 &\scong{16} 4r_{1,4}\;, \label{eqn:r1-repro} \\
ac+bd &\scong 8 2(r_{1,3} + r_{1,4} - r_{3,4})\;. \label{eqn:acbd-again-repro}
\end{align}
We have analogous congruences for the set $\{s_1,s_2,s_3,s_5\}$ by substituting $s_5$ for $s_4$.  Three of these are new, but we only need these two:
\begin{align}
f^2+g^2+h^2 &\scong{16} 4r_{1,5}\;, \label{eqn:r15-repro} \\
af+bg &\scong 8 2(r_{1,3} + r_{1,5} - r_{3,5})\;, \label{eqn:afbg-again-repro}
\end{align}
and moreover, $f\in\oddrats$, \ $g = v\sqrt n$ for some $v\in\oddrats$, and $h^2\in\rats$ with $h^2 \scong 8 0$.  It follows that $bg = tvn$ and $dg = uvn$ are both in $\oddrats$.

We have one final new congruence relation by considering the edge $\{s_4,s_5\}$:
\begin{align}
(f-c)^2 + (g-d)^2 + (h-e)^2 &\scong{16} 4r_{4,5}\;. \label{5points3D4} 
\end{align}
In \erefs{eqn:r1-repro}--(\ref{5points3D4}), $r_{i,j} := ((J_{i,j}/J) \bmod 4) \in \{1,3\}$ for $1\le i<j\le 5$.

Substituting \erefs{eqn:r1-repro} and (\ref{eqn:r15-repro}) into \eref{5points3D4} gives
\begin{align}
2cf + 2dg + 2eh &\scong{16}  4r_{1,4} + 4r_{1,5} - 4r_{4,5} \nonumber \\
cf + dg + eh &\scong{8}  2(r_{1,4} + r_{1,5} - r_{4,5})\;, \label{5points3D6}
\end{align}
and since $cf + dg = cf + uvn \in\rats$, it follows that $eh\in\rats$.
%
%
%
%
Since $e^2 \scong{8} 0 \scong{8} h^2$, we have $e^2 = 2^m w$ and $h^2 = 2^p x$ for some $w,x \in \oddrats$ and integers $m,p \geq 3$.  Therefore, $eh=\pm 2^{(m+p)/2}\sqrt{wx}$.  But since $wx\in\oddrats$ and $eh\in\rats$, it must be that $m+p$ is even and $\sqrt{wx}\in\oddrats$.  Then because $(m+p)/2\ge 3$, we have $|eh|_2 \le 2^{-3} = 1/8$, i.e., by Lemma~\ref{lem:basic-squiggle-congruence}(\ref{item:zero}.),
\begin{equation}\label{eqn:eh}
eh \scong 8 0\;.
\end{equation}

Using \eref{eqn:eh}, we cast \erefs{eqn:acbd-again-repro}, (\ref{eqn:afbg-again-repro}), and (\ref{5points3D6}) with respect to $\scong{4}$ to get
\begin{align}
ac + bd \scong 4 2\;, \label{5points3D8} \\
af + bg \scong 4 2\;, \label{5points3D9} \\
cf + dg \scong 4 2\;. \label{5points3D10}
\end{align}
Furthermore, since all the terms on the left-hand sides are in $\oddrats$, each of them has residue $1$ or $3 \pmod{4}$.  So \erefs{5points3D8}--(\ref{5points3D10}) are only satisfied if the terms in each equation have the same residue (mod~$4$).  Therefore,
\begin{align}
ac &\scong{4} bd = tun\;, \label{5points3D11} \\
af &\scong{4} bg = tvn\;, \label{5points3D12} \\
cf  &\scong{4} dg = uvn\;. \label{5points3D13}
\end{align}
Multiplying \erefs{5points3D11} and (\ref{5points3D12}), we get
\begin{align}
a^2cf \scong 4 b^2dg = (tn)^2uv \label{5points3D14}
\end{align}
Since both $a$ and $tn$ are in $\oddrats$, so are their squares, and $a^2 \scong 4 (tn)^2$ by Fact~\ref{fact:power-of-2}.  Thus \eref{5points3D14} is equivalent to 
\begin{align}
cf \scong 4 uv \label{5points3D15}
\end{align}
So from \erefs{5points3D13} and (\ref{5points3D15}), $uvn \scong 4 uv$, which implies $n \scong 4 1$, contradicting the fact that $n \scong 4 3$ by \eref{eqn:n-is-3}.  Thus no set of five points in $\reals^3$ can be isq-adequate for $U_5$.
\end{proof}

\begin{Cor}
There is no set of five points in $\reals^3$ such that the squares of all interpoint ratios are in $\oddrats$.
\end{Cor}

\begin{Cor}
No arrangement of $n\ge 5$ points in $\reals^3$ is isq-adequate for $U_n$.
\end{Cor}

\section{Conclusions and Further Work}
\label{sec:conclusions}

We have exactly characterized which sets of couplings are adequate for $U_n$ and have given examples of spatial arrangements of qubits adequate for $U_n$ where nearer qubits are more strongly coupled than those farther away.  We have shown planar and $3$-dimensional arrangements of four identical qubits adequate for $U_4$ satisfying the inverse square law and have characterized such arrangements.  We have also shown that this is the best possible; there are no such $n$-point arrangements in three dimensions, for $n\ge 5$.  We have not investigated couplings satisfying other inverse power laws.  We have also not considered arrangements of \emph{non}-identical qubits satisfying the inverse square law (or other inverse power laws).  In this latter case, there would be a constant $c_i$ associated with the $\ordth{i}$ qubit such that $J_{i,j} = c_ic_j/r^d$, where $r$ is the distance separating qubits~$i$ and $j$ and $d$ is a positive constant.

The results of Section~\ref{sec:inverse-square-law} only rules out exact implementations of $U_n$ and not approximate implementations.  It would be interesting to see how useful and feasible approximate implementations would be.

We have concentrated on implementing the operator $U_n$, which is constant-depth equivalent to fanout.  Studying $U_n$ instead of $F_n$ has two theoretical advantages over $F_n$: (1) $U_n$ is represented in the computational basis by a diagonal matrix; (2) unlike $F_n$, which has a definite control and targets, $U_n$ is invariant under any permutation of its qubits, or equivalently, it commutes with the SWAP operator applied to any pair of its qubits.  Are there other operators that are both constant-depth equivalent to fanout and implementable by a simple Hamiltonian?

The Hamiltonian $H_n$ only considers the $z$-components of the spins.  In Heisenberg interactions, the $x$- and $y$-components should also be included in the Hamiltonian, so that a pairwise coupling between spins~$i$ and $j$ would be $J_{i,j}(X_iX_j+Y_iY_j+Z_iZ_j)$.  In \cite{FZ:heisenberg} it was shown that these Hamiltonians can also simulate fanout provided all the pairwise couplings are equal.  We believe we can relax this latter restriction for these Hamiltonians as well.

Finally, the time needed to run our Hamiltonian is inversely proportional to the fundamental coupling constant $J$.  If $J$ is small relative to the actual couplings in the system, then this gives a poor time-energy trade-off and will likely be more difficult to implement quickly with precision.

\appendix

\section{The Quantum Circuit for Parity}
\label{sec:parity}

In this section, we show by direct calculation that the circuit $C_n$ shown in Figure~\ref{fig:parity-circuit} implements the parity gate $P_n$, for all $n\ge 1$.  The special case for $n\equiv 2\pmod{4}$ was shown in \cite{Fenner:fanout}.  Here, $U_n$ is defined by \eref{eqn:U-n}, and
\[ G_n := S^{1-n} = \begin{cases}
    S & \mbox{if $n\equiv 0 \pmod{4}$,} \\
    I & \mbox{if $n\equiv 1 \pmod{4}$,} \\
    S^\dagger & \mbox{if $n\equiv 2 \pmod{4}$,} \\
    Z & \mbox{if $n\equiv 3 \pmod{4}$,}
\end{cases} \]
where $S$ is the gate satisfying $S\ket{b} = i^b\ket{b}$ for $b\in\{0,1\}$, $I$ is the identity, and $Z$ is the Pauli $z$-gate.  ($G_n$ is chosen so that $G_n\ket{b} = i^{b(1-n)}\ket{b}$.)

Fix any $x_1,\ldots,x_n,t\in\{0,1\}$.  For convenience, we separate the first $n-1$ qubits, which only participate in $U_n$ and $U_n^\dagger$, letting $\vec x := x_1\ldots x_{n-1}$.  We set $w := w(\vec x) = x_1+\cdots+x_{n-1}$ and $W := w + x_n$, the Hamming weight of $x_1\cdots x_n$.  We set $p := W\bmod 2$, the parity of $x_1\cdots x_n$, which will be XORed with $t$ in the target qubit.  Running the circuit starting with initial state $\ket{\vec x}\ket{x_n}\ket{t}$, we have
\begin{align*}
    \ket{\vec x}\ket{x_n}\ket{t}
    &\stackrel{H}{\longmapsto} 2^{-1/2}\ket{\vec x}\left(\ket{0} + (-1)^{x_n}\ket{1}\right)\ket{t} \\
    &= 2^{-1/2}\left(\ket{\vec x,0} + (-1)^{x_n}\ket{\vec x,1}\right)\ket{t} \\
    &\stackrel{U_n}{\longmapsto} 2^{-1/2}\left(i^{w(n-w)}\ket{\vec x,0} + (-1)^{x_n}\,i^{(w+1)(n-w-1)}\ket{\vec x,1}\right)\ket{t} \\
    &= 2^{-1/2}\,i^{w(n-w)}\ket{\vec x} \left(\ket{0} + i^{n-1-2(w+x_n)}\ket{1}\right)\ket{t} \\
    &= 2^{-1/2}\,i^{w(n-w)}\ket{\vec x} \left(\ket{0} + (-1)^W\,i^{n-1}\ket{1}\right)\ket{t} \\
    &\stackrel{G_n}{\longmapsto} 2^{-1/2}\,i^{w(n-w)}\ket{\vec x} \left(\ket{0} + (-1)^W\ket{1}\right)\ket{t} \\
    &= 2^{-1/2}\,i^{w(n-w)}\ket{\vec x} \left(\ket{0} + (-1)^p\ket{1}\right)\ket{t} \\
    &\stackrel{H}{\longmapsto} i^{w(n-w)}\ket{\vec x} \ket{p}\ket{t}\;.
\end{align*}
At this point, the C-NOT gate is applied, resulting in the state $i^{w(n-w)}\ket{\vec x} \ket{p}\ket{t\oplus p}$.  The remaining gates undo the above action on the first $n$ qubits, resulting in the state $\ket{\vec x}\ket{x_n}\ket{t\oplus p}$, which is the same as $P_n$ applied to the initial state.

\bigskip

Finally, we note that $C_n$ only depends on $U_n$ up to an overall phase factor: any gate $V_n\propto U_n$ can be substituted for $U_n$ in the circuit, because any phase factor introduced by applying $V_n$ on the left will be cancelled when $V_n^\dagger$ is applied on the right.  This fact is, of course, unnecessary for physical implementation.

\section{Implementing $U_n$ with Equal Couplings: Proof of Lemma~\ref{lem:U-n}}
\label{sec:U-n}

In this section give an updated proof of Lemma~\ref{lem:U-n}, which we restate here:

\begin{Lem}
For $n\ge 1$, let $H_n := J\sum_{1\le i<j\le n} Z_i Z_j$ for some $J>0$.  Then $U_n = V_n(t,\theta)$ for some $\theta\in\reals$, where $t := \pi/(4J)$ and $V_n(t,\theta)$ is as in \eref{eqn:V-n}.
\end{Lem}

\begin{proof}
Looking at Equations~(\ref{eqn:U-n}) and (\ref{eqn:V-n}), we see that for $t,\theta\in\reals$, the condition $V(t,\theta) = U_n$ is equivalent to
\[ \exp\left(-i\theta-i\sum_{1\le i< j\le n} Jt(-1)^{x_1+x_j}\right) = i^{w(x)(n-w(x))}
\]
holding for all $x\in\{0,1\}^n$.  Noting that $i = e^{i\pi/2}$ and $Jt = \pi/4$ and equating exponents, this condition becomes
\begin{equation}\label{eqn:phase-equal-couplings}
\theta + \frac{\pi}{4}\sum_{1\le i<j\le n} (-1)^{x_i+x_j} \equiv_{2\pi} -\left(\frac{\pi}{2}\right)w(x)(n-w(x))
\end{equation}
for all $x\in\{0,1\}^n$ (cf.\ Equations~(\ref{eqn:V-n-of-x}) and (\ref{eqn:U-n-phase})).
The sum on the left-hand side becomes
\begin{align*}
\sum_{i<j} (-1)^{x_i + x_j}
&= \frac{1}{2}\sum_{i\ne j} (-1)^{x_i + x_j} = -\frac{n}{2} + \frac{1}{2}\sum_i\sum_j (-1)^{x_i+x_j} = -\frac{n}{2} + \frac{1}{2}\left(\sum_{i=1}^n (-1)^{x_i}\right)^2 \\
&= -\frac{n}{2} + \frac{1}{2}\left(\sum_i (1-2x_i) \right)^2 = -\frac{n}{2} + \frac{1}{2}\left(n-2w(x)\right)^2
= \frac{n^2-n}{2} - 2w(x)(n-w(x))\;.
\end{align*}
Substituting this back into \eref{eqn:phase-equal-couplings} satisfies it, provided we set $\theta := \pi(n^2-n)/8$.
\end{proof}

\end{document}